\documentclass{lmcs} 
\keywords{Lévy-Prokhorov metric,
behavioural distance,
epsilon-bisimulation,
reactive probabilistic transition systems,
discrete labelled Markov processes,
coalgebraic epsilon-(bi)simulation} 

\usepackage{hyperref}
\theoremstyle{plain} %\crefname{satz}{Satz}{S\"atze}
\usepackage{graphicx,xcolor,wrapfig}
\usepackage{amssymb,amsmath}
\usepackage{amsmath,latexsym}
\RequirePackage{stmaryrd}

 \usepackage[T1]{fontenc}
\usepackage[utf8]{inputenc}

\definecolor{darkross}{rgb}{0.016,0.404,0.667}   % blue of Ross Seal
\definecolor{middleross}{rgb}{0.020,0.522,0.859}
\definecolor{lightross}{rgb}{0.094,0.624,0.980}
\definecolor{blue}{named}{middleross}

\usepackage[all]{xy}
\CompileMatrices

\def\ang#1#2{\langle#1\rangle_{#2}\, }

\def\cA{\mathcal{A}}
\def\cD{\mathcal{D}}

\def\cR{R}
\def\LP{\mathrm{LP}}

\def\Sets{\mathbf{Sets}}
\def\PMet1{\mathbf{PMet_1}}

\renewcommand{\epsilon}{\varepsilon}

\newtheorem{theorem}{Theorem}
\newtheorem{lemma}[theorem]{Lemma}
\newtheorem{proposition}[theorem]{Proposition}
\newtheorem{corollary}[theorem]{Corollary}

\newtheorem{definition}[theorem]{Definition}
\newtheorem{example}[theorem]{Example}
\newtheorem{remark}[theorem]{Remark}

\begin{document}

% If the title is longer than 55 characters, then specify a shorter running title as the optional argument to \title. The running title should be roughyl at most 55 characters:
\title{$\epsilon$-Distance via Lévy-Prokhorov Lifting}
%\title[Instructions]{Instructions for Authors\texorpdfstring{\\}{. }How to prepare papers  for LMCS using \texorpdfstring{\MakeLowercase{\texttt{lmcs.cls}}}{lmcs.cls}\rsuper*\\

%\titlecomment{{\lsuper*}OPTIONAL comment concerning the title, \eg,  if a variant or an extended abstract of the paper has appeared elsewhere.}
\thanks{This work was partly done during a sabbatical of Jos\'ee Desharnais at the University of Salzburg, as well as during Dagstuhl Seminar 24432 and the Bellairs Workshop on Quantitative Reasoning 2025. We thank these venues, the organizers, and the participants for providing a perfect working environment. 
 We also thank Matteo Mio and Franck van Breugel for some fruitful and motivating discussions, and the anonymous reviewers for their insightful and constructive suggestions.}	%optional

% If you are submitting to a special issue, you do not need to fill out this
% field (although you may).  Please be aware, however, that upon publication
% the layout editing team will insert the full special issue name, which may
% push down the text of your first section.
\specialissue{CSL}

% affiliations are numbered automatically with a, b, c (see below)
% use the optional argument to indicate the affiliation(s) of each author
% omit the argument if there is only one author, or only one affiliation
\author[J.~Desharnais]{Josée Desharnais\lmcsorcid{0000-0003-2410-3314}}[a]

\author[A.~Sokolova]{Ana Sokolova\lmcsorcid{0000-0002-8384-3438}}[b]
% affiliation 1 (automatically numbered a)
\address{Laval University, Canada}	%optional
% write emails for all authors having that affiliation
\email{josee.desharnais@ift.ulaval.ca}  %optional

% affiliation 2 (automatically numbered b)
\address{University of Salzburg, Austria}	%optional
\email{ana.sokolova@cs.uni-salzburg.at}  %optional

% mandatory, please use full name; only 1 author per \author macro; first two parameters are mandatory, other parameters can be empty. Please provide at least the name of the affiliation and the country. The full address is optional. Use additional curly braces to indicate the correct name splitting when the last name consists of multiple name parts.

%% required for running head on odd and even pages, use suitable
%% abbreviations in case of long titles and many authors:

%%%%%%%%%%%%%%%%%%%%%%%%%%%%%%%%%%%%%%%%%%%%%%%%%%%%%%%%%%%%%%%%%%%%%%%%%%%

%% the abstract has to PRECEDE the command \maketitle:
%% be sure not to issue the \maketitle command twice!

\begin{abstract}
The most studied and accepted pseudometric for probabilistic processes is one based on the Kantorovich distance between distributions. It comes with many theoretical and motivating results, in particular it is the fixpoint of a given functional and defines a functor on (complete) pseudometric spaces. It is also the foundation for a categorical lifting of pseudometrics.

Other notions of behavioural pseudometrics have also been proposed, one of them (\texorpdfstring{$\epsilon$}{epsilon}-distance) based on  \texorpdfstring{$\epsilon$}{epsilon}-bisimulation. \texorpdfstring{$\epsilon$}{epsilon}-Distance has the advantages that it is intuitively easy to understand, it 
relates systems that are conceptually close (for example, an imperfect implementation is close to its specification), and it comes equipped with a natural notion of \texorpdfstring{$\epsilon$}{epsilon}-coupling. Finally, this distance
is easy to compute.

We show that \texorpdfstring{$\epsilon$}{epsilon}-distance is also the greatest fixpoint of a functional and provides a functor. The latter is obtained by replacing the  Kantorovich distance in the lifting functor with the Lévy-Prokhorov distance. In addition, we show that \texorpdfstring{$\epsilon$}{epsilon}-couplings and $\epsilon$-bisimulations have an appealing coalgebraic characterization. 
\end{abstract}

\maketitle

%% start the paper here:  END OF LMCS heading

%\section*{Todo or to think}
%Do we change seciton 1 and 2?  YES, AT LEAST A BIT.

%all todo are done
%\textcolor{blue}{sub-distribution vs. subdistribution, done
%wrt. vs w.r.t., done $S, X, A,$ in lemma 40 and 41 done
%marginpar done
% errors \jo color} done

\section{Introduction} \label{sec:intro}

Probabilistic systems~\cite{Giacalone89,Larsen91,SL94:concur,Bai98,Baier2008,Panangaden09} 
and their behaviour have been the object of study of over thirty years in the area of formal verification and analysis of systems. They are used to represent uncertainty, incomplete information, as well as randomized behaviour. One important direction, in order to prove that systems behave the same, is the study of behavioural equivalences: these identify states with (exactly) the same behaviour. Behavioural equivalences are a useful way to analyse and compare behaviour of systems, they have solid foundations in concurrency theory~\cite{Mil89a,GlabbeekSS95} as well as elegant abstract generalizations in the theory of coalgebras~\cite{deVink99,Sokolova05,Hasuo06,Jacobs0S15,Hasuo10,Jacobs2016-book,BSS17,BonchiSV22}.
However, as was already observed in ~\cite{Giacalone90}, in a non-exact world, e.g., when the probabilities in the model are approximate and not exactly known, or estimated by sampling, or when small differences should not be considered the same as large differences, 
behavioural equivalences may be too strong. 

One solution is to employ a distance, actually a pseudometric, providing a quantitative notion of how much states in a probabilistic model differ from one another, or how close they are to each other--distance zero corresponding to equivalence. The study of pseudometrics~\cite{Desharnais99b,Worrell01,Worrell01b,Desharnais04,AlfaroFS04,breugel05,breugel07,DesLavTra08,Worrell01,Worrell01b,TraQapl11,10.1145/3157831.3157837} and quantitative theories~\cite{MardarePP16,MardarePP17,Bacci0LM17,BacciBLM17,BaldanBKK18,DBLP:journals/logcom/SprungerKDH21,MioSV21,DBLP:conf/fossacs/GoncharovHNSW23,DAngeloGK0NRW24,MioSV21,MioSV24} has been a fruitful one in the past decade(s), the references above showing only some of the original sources.
 
Historically, the impulse to define a behavioral distance between probabilistic processes became imperative for those with a \emph{continuous state space and continuous probability distributions}, called Labelled Markov Processes (LMPs). It started with the idea of defining a real-valued logic, as proposed by Kozen~\cite{Kozen81}, to express properties satisfied by processes: Taking the supremum over the differences of the values of formulas was a natural way of defining a pseudometric, see Desharnais et al.~\cite{Desharnais99b,Desharnais04}. The authors later on observed that this distance was the greatest fixpoint of some functional~\cite{Desharnais02b} and, in between and later, van Breugel and Worrell~\cite{Worrell01,Worrell01b,breugel05} and van Breugel et al.~\cite{breugel07}, proved that this gives a functor (monad) that is a lifting of the distribution functor (monad) using the Kantorovich metric on distributions. 
%To our knowledge, no other such metric has been lifted to transition systems in a categorical way. 
Moreover, the distance is the one obtained using final coalgebra semantics. Since then, this pseudometric has been studied further in different contexts and algorithms have been proposed to compute it~\cite{TangB16, TangB17,Bacci0LM17,BacciBLM17}. 
   The distance comes with many theoretical and motivating results, and has also been extended by Baldan et al.~\cite{BaldanBKK18} and Sprunger et al.~\cite{sprunger18,DBLP:journals/logcom/SprungerKDH21} to a generic Kantorovich-style (and dually Wasserstein) lifting of arbitrary functors parametric in certain evaluation maps---the two liftings coincide on distributions. 
%\begin{itemize}
%\item defined as the highest difference a real valued logical formula captures in the behavior of the processes~\cite{Desharnais04}
%\item defined as the fixpoint of a functional on pseudometrics~\cite{Desharnais04,breugel05}
%\item the first defined, using a real valued logic that is very intuitive~\cite{}
%\item defined as the distance in a terminal coalgebra~\cite{breugel05}.
%\end{itemize}
With all these seals of approval, the Kantorovich behavioural distance has imposed itself as the one to use and study. One aspect that it lacks  is an intuitive accompanying notion of approximate bisimulation.

Approximate bisimulations provide another strategy to circumvent the non-robustness of bisimulation equivalences.  They have been studied in non-probabilistic systems~\cite{ Ying00,Pappas07}, and in probabilistic   systems~\cite{DesLavTra08,TraQapl11,DInnocenzoAK12, ABATE2013, AbateKNP14, BianA17, SBKPQ2024}. 
%note: BianA17 use uncountable Lmcs. they use only part of the theory. e.g. no link with the logic, and bisim is not 0-bisim.
%AlfaroFS04 : study distances but no approx bisim
We focus on the notion of $\epsilon$-bisimulation defined in~\cite{DesLavTra08}.
One aim of such approximate bisimulations is to unite the best of both worlds:  give a relational structure for reasoning about systems and at the same time define a distance. States related by an $\epsilon$-bisimulation are at distance at most $\epsilon$ from one another.
The advantages of such a distance  are its
intuitive nature and its ability to 
relate systems that  are almost the same in structure, as the next example will illustrate.  
Finally,  the distance can be computed relatively easily.

For comparison with the Kantorovich behavioural distance, consider the imperfect channel (when $\epsilon > 0$) depicted in Figure~\ref{fig:channel}.
The $\epsilon$-distance of this channel to a perfect channel 
\begin{figure} % 15cm is usually too wide to have the figure next to text
    $\xymatrix@C=13mm@R=10mm{
     c_\epsilon \ar@/^1pc/[rr]^{\textit{send} ~[1-\epsilon]} 
 && 1 \ar@/^1pc/[ll]^{ack[1]}}$
    \caption{A simple channel that drops a token with probability $\epsilon$}
    \label{fig:channel}
    \end{figure}
(one with $\epsilon = 0$) is $\epsilon$, and it provides a simple example that  imperfect implementation can be considered close to its specification. On the other hand, the Kantorovich distance (without discount) gives distance 1 to the pair of perfect and imperfect channels, which  reflects the fact that the two channels have a very different behavior on the long run--see Example~\ref{ex:channel-K}  for more details.
Interestingly, 
$\epsilon$-distance has an almost build-in notion of continuity: Two  channels with close values of $\epsilon$ are close in the $\epsilon$-distance as well. 

However, the approximate equivalences had no categorical formulation attached and it was not clear whether the associated distance shares some  of the generic good properties identified for the Kantorovich metric.  
Our main motivation for this work was to investigate whether we can remedy this situation. Indeed, we provide the missing abstract characterization of $\epsilon$-bisimulation in two directions: by obtaining a fixpoint characterization of the induced distance and by giving a coalgebraic view following Aczel-Mendler coalgebraic bisimilarity. The fixpoint characterization is the main result of the paper. For this, we show how to replace the Kantorovich pseudometric on distributions with the \emph{L\'evy-Prokhorov pseudometric} on distributions, and show that the distance obtained is indeed the $\epsilon$-distance. It is remarkable how well the L\'evy-Prokhorov distance fits into the definition of $\epsilon$-bisimulation: We show that $\epsilon$-distance is the greatest fixpoint of a suitable functional, and that the lifting also lifts the discrete distributions functor to a functor on pseudometric spaces. We also show a result that has no matching for the Kantorovich functional: any (post)fixpoint distance of the L\'evy-Prokhorov functional %defines $\epsilon$-balls around states that contain $\epsilon$-bisimilar states.%or just : 
defines an $\epsilon$-bisimulation.
Applications of our results may take different routes: 
On the practical side, the distance could be computed by an iterative algorithm as we show that Kleene's fixpoint theorem applies.
On the theoretical side, the fixpoint characterization may open the way for studying new problems, one example being $\epsilon$-bisimulations up-to as an acceleration to computing $\epsilon$-bisimilarity, along the lines of~\cite{BonchiKP23}.   

Our observation may also open a new way of studying other pseudometrics on probabilistic systems that would be constructed from other basic distances on distributions.  %More generally, the construction could be applied to a  pseudometric on $FX$ for a functor $F$ on $\Sets$. %%% It is yet to be seen whether the L\'evy-Prokhorov lifting fits  the Kantorovich / codensity - style liftings of Baldan et al.~\cite{BaldanBKK18} which aim at generalizing the Kantorovich lifting to generic functors. If not, perhaps the lifting can provide another way of lifting functors to metric spaces. --- NOW SEEN ! ;-)

% LMPs~\cite{Desharnais04,Worrell01} and some work has been done to estimate bisimulation with a certain degree of confidence~\cite{Ferns05,DesLavZhi06}.  %, FSTTCS}.  

Unlike in the case of the Kantorovich lifting, the L\'evy Prokhorov lifting does not yield a monad on the category of pseudometric spaces, see Section~\ref{sec:axiom}. As a consequence, we could not follow the path of~\cite{breugel07} and reproduce their results on the final coalgebra. Instead, we take a different way and describe $\epsilon$-(bi)simulations as coalgebraic (bi)simulations, following and extending the framework of Hasuo~\cite{Hasuo06,Hasuo10} and the original results of Hughes and Jacobs~\cite{HughesJacobs04}.
%todo : Ana, this is not clear to me, can you improve?
We only present the necessary abstraction needed to discuss $\epsilon$-(bi)simulations on discrete time Markov chains and labelled Markov processes, generalizations of these notions are possible and we plan to elaborate on them in follow-up work. 

The paper is an extended version of~\cite{DBLP:conf/csl/DesharnaisS26}. The extension consists of all proofs, an iterative characterisation of the fixpoint, a brief discussion of discounting the future, as well as more elaborated examples. Within this year, this work has already inspired several interesting follow-up results~\cite{DBLP:conf/fossacs/WildSMKF26,DBLP:conf/concur/ForsterSWKN26} that we mention in context below.

\section{The objects of interest}			%------------------------------------------------------

Our development concerns discrete labelled Markov processes (LMPs) also called reactive probabilistic transition systems. They consist of a set of states $S$ and transitions $\tau_a: S\to \cD S$ labelled with actions $a$ of some set $\cA$, where $\cD S$ is the set of discrete subdistributions on $S$:
$$\cD S =\{ \varphi \colon S \to [0,1] \mid \sum_{s \in S }\varphi(s) \le 1 \}.$$
Note that the size of the state set need not be restricted (for the sums to be defined), as the sum is defined as $\sum_{x \in S} \varphi(x) = \sup_{{\textrm{finite }}X \subseteq S} \left(\sum_{x \in X} \varphi(x)\right) $. It is easy to prove that discrete probability distributions always have a countable support, i.e., they assign non-zero probability to at most countably many elements of $S$. 

For  $s,t\in S$ and $a\in \cA$, the value $\tau_a(s)(t)$ (also written $\tau_a(s,t)$) encodes the probability of jumping from $s$ to $t$ when action $a$ is taken.  Figure~\ref{fig:channel} shows such a system with actions \emph{send} and \emph{ack}. The first transition is a subprobability, where the missing probability represents the probability of failure.
If $X\subseteq S$, we write $\tau_a(s)(X)$ for $\sum_{t\in X} \tau_a(s)(t)$.
Coalgebraically, LMPs, for a set of actions or labels $\cA$, are arrows $\tau\colon S \to (\cD S)^\cA$. We write $\tau_a(s)$ for the subdistribution $\tau(s)(a)$. 
In case of just one label, LMPs are (discrete time) Markov chains, i.e., coalgebras of the functor $\cD$. The reason for working with subdistributions is that they come equipped with a non-trivial pointwise order. 
Bisimilarity of LMPs is the classical Larsen-Skou~\cite{Larsen91} bisimilarity. 

We denote by $\PMet1$ the category of $1$-bounded pseudometric spaces. It has as objects pseudometric spaces, which are pairs $\langle X,d \rangle$ of a set $X$ and a function $d \colon X\times X \to [0,1]$ which satisfies  $d(x,x) = 0$, is symmetric: $d(x,y)= d(y,x)$ and satisfies the triangular inequality: $d(x,y) + d(y,z) \ge d(x,z)$ for all $x,y,z \in X$. As arrows, $\PMet1$ has non-expansive maps $f\colon \langle X,d_X\rangle \to \langle Y,d_Y\rangle$, i.e.,  functions $f \colon X \to Y$ satisfying $d_Y(f(x_1),f(x_2)) \le d_X(x_1,x_2)$. A metric or distance is a pseudometric that additionally satisfies $d(x,y) = 0$ iff $x=y$. Of particular interest for LMPs are pseudometrics whose kernel is bisimilarity, that is, $d(x,y) = 0$ iff $x$ and $y$ are bisimilar.  As we only deal with pseudometrics here, we often use the word distance or metric as a shorthand for pseudometric.
Following~\cite{breugel05,breugel07}, we will  use the opposite pointwise  order\footnote{One reason for this choice is to obtain the distance as the greatest, rather than least, fixpoint -- just like bisimilarity is the greatest fixpoint of a suitable functional.} on distances defined by: $$d_1 \sqsubseteq d_2 \quad \Longleftrightarrow \quad \forall x,y \in X : d_1(x,y) \ge d_2(x,y).$$
However, whenever it is clearer, we may use the direct pointwise order too:  $ d_1\geq d_2 \Leftrightarrow d_1 \sqsubseteq d_2 $.

 \subsection{The Kantorovich pseudometric} %%%%%%%%%%%%%%%%%%%
The observation that bisimulation is too strong in the context of probabilistic systems has led to the idea of defining a pseudometric that would give zero distance to bisimilar states. Initially, this pseudometric was defined using a real-valued logic. A set $\mathcal F$ of functionals were defined from states of LMPs to $[0,1]$, with the following syntax, mimicking logical formulas~\cite{Desharnais04}:
$$
f:= 1\mid  \inf (f_1, f_2) \mid 1-f \mid  f\varominus q\mid \ang{a}{}f \mbox{ with } q\in[0,1]\cap\mathbb{Q}.
$$
We omit the semantics, but the next example will  give a taste of it.
A distance emerged naturally as
$$d_{K}(s,t)=\sup_{f\in \mathcal{F}} \big|f(s)-f(t)\big|.
$$
\begin{example}\label{ex:channel-K}
    As an example, we can look back at the channel of Figure~\ref{fig:channel}. The maximum difference over functionals between $c_\epsilon$ and $c_\gamma$, with $\epsilon,\gamma>0$ is over the  functional $\langle \textit{send\,} \rangle\langle \textit{ack}\, \rangle 1 $, which evaluates to $1-\epsilon$ and $1-\gamma$, respectively,  and thus yields the distance  $d_K(c_\epsilon,c_\gamma) = |\epsilon-\gamma|$. 
    Between $c_\epsilon$ and $c_0$, the functional $(\langle send\rangle\langle ack \rangle)^n 1 $ evaluates to $(1-\epsilon)^n$ on $c_\epsilon$, and to 1 on $c_0$. So the supremum results in $d_K(c_\epsilon,c_0)=1$.
%the semantics of <a>f
    % of $\ang{a}{\!\!}f$  is the only one we will use. For $s\in S$, on an LMP  $\tau:S\to(\cD S)^\cA$,  it evaluates to $\ang{a}{\!\!}f(s) = \sum_{t\in S} f(t)\tau_a (s,t)$.
\end{example}
Later on, it was proven that the functionals could be any non-expansive maps and that this distance is the distance obtained via a final coalgebra construction for a functor (monad) that is a lifting of the Distribution functor (monad) to pseudometric spaces using the Kantorovich metric, by van Breugel and Worrell~\cite{breugel05} and van Breugel et al.~\cite{breugel07,breugelSW07}.

This distance is
a fixpoint of the functional on distances on the states of an LMP  $\tau:S\to(\cD S)^\cA$
\begin{equation}
    \Delta_{K}(d)(s_{0},s_{1})= \sup_{a \in \cA}\,\delta_K^d(\tau_{a}(s_{0}),\tau_{a}(s_{1})) 
    \label{eq:DeltaK}
\end{equation}
where $\delta_K^d$ is the well-known Kantorovich distance between two distributions. We do not need its definition here, the interested reader can find it in, e.g., \cite{breugel07}.
%\item \emph{finite approximants} converge in $d^c_{K}$ {\small \color{gray}($c$ is a discount factor)}
%\item   {algorithms}  {\scriptsize\color{gray}[van Breugel, Sharma, Worrell,  Tang, Bacci, Bacci, Larsen, Mardare]} 

The distance on distributions is just a parameter in Equation~\eqref{eq:DeltaK}; changing the Kantorovich distance $\delta_K^d$ to another distance on distributions and looking for a fixpoint will give us a new distance on states. 
In fact we are interested in the question formulated the other way around: find a distance on distributions for which the fixpoint is the $\epsilon$-distance $d^*$, defined in the next section. 
The functional $\Delta_K$ was  discovered  after the behavioural distance $d_K$ was introduced as well. However, it could have been done the other way around, choosing a distance on distributions, and looking for the resulting distance on states, similarly as how it is done by Baldan et al.~\cite{BaldanBKK18}. We discuss this in more detail again in Section~\ref{sec:LP} below.

\subsection{$\epsilon$-Bisimulation}    %%%%%%%%%%%%%%%%%%%%%%%%
Bisimulation being too strong for the comparison of states in quantitative systems, has also led to relaxing the definition of simulation and bisimulation itself to approximate relations. As mentioned in the introduction, there are a few of those definitions of approximate bisimulations, but we are interested in the following.

\begin{definition}[\cite{DesLavTra08}]
\label{d:eSim}
Let  $\tau:S\to(\cD S)^\cA$ be an LMP and let $\epsilon\in[0,1]$. A relation $R\subseteq S\times S$ is an \emph{$\epsilon$-simulation} if whenever  $s R t$, 
then for all $a\in \cA$, for all $ X\subseteq S$
$$  \tau_a(s)(X)\leq \tau_a(t)({R}(X))+\epsilon, \qquad\mbox{ where } R(X) = \{y\mid \exists x\in X : x R y\}.
$$
If $R$ is symmetric, it is an \emph{$\epsilon$-bisimulation}. \\ A state
 $s$ is $\epsilon$-simulated by state $t$, written \emph{$s\prec_\epsilon t$}, if $s R t$ for some  $\epsilon$-simulation $R$. \\ If $s R t$ for $R$ $\epsilon$-bisimulation, we write $s\sim_\epsilon t$ and we say that $s$ and $t$ are $\epsilon$-bisimilar.
\end{definition}
As expected, ordinary (bi)simulation on LMPs~\cite{Larsen91,Blute97} is simply $0$-(bi)simulation. This definition has an extension to nondeterministic and probabilistic finite systems~\cite{TraQapl11}.

The operation $\cR(X)$ on a set $X$ is what restricts this work to discrete distributions. In general, if $X$ is measurable, we may not have $R(X)$ measurable. Working in analytic spaces, as was done before for LMPs~\cite{Desharnais00} may solve this issue, but we leave it for future work. In particular, one could ask  $R(X)$ to be measurable whenever $X$ is, as done in~\cite{BianA17}, but many results are not proven in that case, like the logical characterisation, or Theorem~\ref{t-kernel d} below.

\begin{example} 
\label{ex:e-bisim}
Consider the following example, with $\gamma \in (0,1]$. 
$$
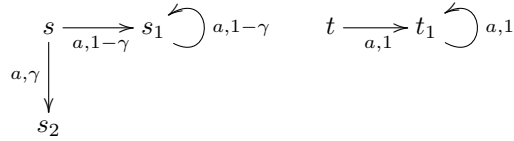
\xymatrix{
s\ar[d]_{a,\gamma}\ar[r]_{a,1-\gamma} & s_1\ar@(dr,ur)[]_{a,1-\gamma} 
&&  t\ar[r]_{a,1} & t_1\ar@(dr,ur)[]_{a,1} \\s_2%&&&\bullet
}
$$
The relation $R=\{(s,t),(s_1,t_1)\}$ is an $\epsilon$-simulation for $\varepsilon = \gamma$, and so is $R\cup R^{-1}$. Hence, $s\prec_{{\gamma}} t$, $t\prec_{{\gamma}} s$
 and $s\sim_{{\gamma}} t$.
 However, the situation is different for  $\epsilon<\gamma$: for example, taking $\epsilon = 0$, we observe that $s\prec_0 t$ (with the relation $R \cup \{ (s_2,t_1)\}$).
However, $t\not\prec_0 s$, because a relation $R'$ relating $t$ and $s$ would have to also relate $t_1$ to $s_1$ and to $s_2$, if $\gamma <1$, and $t_1$ to $s_2$ if $\gamma = 1$.  
Indeed, taking $X = \{t_1\}$, we need $\tau_a(t)(X)\leq \tau_a(s)(R'(X)) +0$, but for this to be the case, we would need $\tau_a(s)(R'(X)) = 1$ and hence $R'$ must include the pairs $(t_1,s_1)$ and $(t_1, s_2)$ if $\gamma < 1$, and otherwise must include $(t_1,s_2)$ if $\gamma = 1$. So, in any case $(t_1, s_2) \in R'$, but the pair $(t_1, s_2)$ cannot be related by any $\epsilon$-simulation relation $R'$ for $\epsilon <1$ (which is the case here as $\epsilon = 0$), as,  for example, taking again $X = \{t_1\}$, we would have $\tau_a(t_1)(X) = 1$ but $\tau_a(s_2)(R'(X)) = 0$.
Hence, also $s\not\sim_0 t$. 

\end{example}
In contrast to the case of 0-bisimilarity, two-way $\epsilon$-similarity is not $\epsilon$-bisimilarity~\cite{DesLavTra08,TraQapl11}. 
From the notion of $\epsilon$-bisimilarity, a pseudometric arises naturally by taking the infimum.

 \begin{definition}[The $\varepsilon$-distance $d^*$~\cite{DesLavTra08}]
Let $\tau:S\to(\cD S)^\cA$ be an LMP. The pseudometric $d^*$ on $S\times S$, called $\varepsilon$-distance, is defined as follows:
$$\begin{array}{lcl}
d^*:&S\times S&\rightarrow[0,1]\\
   &(s,t)&\mapsto \inf\ \left\{ \epsilon\in[0,1]\mid s\sim_\epsilon t\right\}.
\end{array}$$
\end{definition}
The function $d^*$ is a pseudometric on the states of the considered LMP: the triangle inequality comes from the well-known fact, see e.g.~\cite{DesLavTra08}, that $s\sim_{\epsilon_1} u$ and $u\sim_{\epsilon_2}t$ imply $s\sim_{\epsilon_1+\epsilon_2}t$. (The same property holds for $\prec_\epsilon$ as well.)
Indeed, $\prec_\epsilon$ and $\sim_\epsilon$ are not transitive. Instead, these  relations are \emph{entourages} that form a \emph{uniform structure}, or \emph{uniformity}, cf.~\cite{Bourbaki71}.

The kernel of $d^*$, that is, the set of pairs at zero distance, is $\sim_0$, the bisimilarity relation on LMPs.

\begin{theorem}[\cite{DesLavTra08}]\label{t-kernel d}
$d^*(s,t)=0 \Leftrightarrow s\sim_{0}t$.
\end{theorem}

One of the nice properties of this distance is its relatively easy computability, especially regarding computation 
    ``by hand''. It suffices to define an $\epsilon$-bisimulation to obtain an upper bound to the distance between states. In Example~\ref{ex:e-bisim}, we have $d^*(s,t)=\gamma$, and, indeed, we easily see that the probabilities of $s$ and $t$ seen as processes are within $\gamma$. Also here, the Kantorovich distance between these states is 1. One way to see that is by using the functional $\langle a \rangle^n 1 $, which evaluates to $(1-\epsilon)^n$ on $s$, and to 1 on $t$. So $d_K(s,t)=1$.
Similarly, as already noted in the introduction, for the channel in Figure~\ref{fig:channel}, we can prove that $d^*(c_\epsilon,c_\gamma) = |\epsilon-\gamma|=d_K(c_\epsilon,c_\gamma)$, for $\epsilon,\gamma>0$, but $d^*(c_\epsilon,c_0) = \epsilon\neq d_K(c_\epsilon,c_0)$.
In fact, $\epsilon$-distance is incomparable to both the discounted and the undiscounted Kantorovich distance, as argued in~\cite{DesLavTra08}.
That paper also gives a straightforward polynomial algorithm to compute $d^*$.  The computation of the Kantorovich distance, discounted or not, has been the subject of many papers, which also led to   %distance zero and distance 1 can be computed in polynomial time, but for precisions in between, it is costly.
 polynomial algorithms~\cite{TangB17}.% TODO. %I think...! 
%this papers 
%https://www.sciencedirect.com/science/article/pii/S0022000020300209#br0260
%paper tanq and breugel "the discounted probabilistic bisimilarity pseudometric is rational and can be computed exactly by a polynomial-time algorithm"

\section{The L\'evy-Prokhorov distance lifting} \label{sec:LP}
We will now explain in details what we think should be a behavioral metric on LMPs. These details may not be new, or may appear obvious to the expert reader, but we find it useful to spell them out here. 
%\footnote{If the reader knows a reference, we would be grateful.} 
When looking for a distance $d^\dagger$ on the states of an LMP, one  observes 
that states are both targets of distributions and  (a set of) distributions themselves (since they are \textit{defined} by their outgoing transition distributions): that is, in  Example~\ref{ex:e-bisim}, the states $s$ and $t$ can be viewed as distributions over the set $\{s_1,s_2,t_1\}$.  As there are already a few distances on distributions that were studied outside computer science and concurrency theory, they give a starting point for distances on states: we could say, given such a distance~$\delta$:
\begin{equation}
d^\dagger(s,t):= \sup_{a\in\cA}\,\delta(\tau_{a}(s),\tau_{a}(t)),
    \qquad\mbox{ for $s,t\in S$.}
    \label{eq:def_d_naive}
\end{equation}
We will use the letter $d$ for distances on states, and the symbol $\delta$ for distances on distributions. 
Analysing the options of this equation in full generality is outside the scope of the current paper%\footnote{If the reader knows such works, we would be happy to list them here.}
, but some distances, such as $\delta_K^d$ (and the one we introduce below  $\delta^d_{\LP}$) have a particularity. They are defined using a parameter $d$, a basic distance on the space where the distributions are defined. This is desirable for concurrency theory because distributions on states that are different but close (or even bisimilar!) should also be close: so a starting metric $d$ on the state space is of key importance to account for these similarities between states/processes.  Note that a few papers~\cite{Pappas07, Alfaro04} start from metrics on the states that account for extra information (e.g. a distance on the labels/observations that are attached to states)--but this goes beyond the behaviour of states that we want to capture here (although they constitute an interesting line of work to extend our method). 

A needed property for $d^\dagger$ to be a \textit{behavioural} distance, as also pointed out in~\cite{BaldanBKK18} for $\delta_K^d$, 
 is  that it accords with the starting distance $d$ used in  $\delta^d$, as follows.
\begin{equation}\label{eq:d-dagger-fixpoint}
    d^\dagger(s,t) = \sup_{a\in\cA}\,\delta^{d^\dagger}(\tau_{a}(s),\tau_{a}(t)),
    \qquad\mbox{ for $s,t\in S$.}
\end{equation}
This says that the distance between states viewed as simple members of the space (the left hand part of the equality), is the same as their distance when viewed as processes, that is, their outgoing transition distributions. 
 So this really says that $d^\dagger$ treats states according to their behaviour. 
Technically, this is saying that we are looking for a distance fixpoint $d^\dagger$ of the functional
\begin{equation}
    \Delta(d)(s,t):= \sup_{a\in\cA}\,\delta^d(\tau_{a}(s),\tau_{a}(t)), 
    \qquad\mbox{ for $s,t\in S$.}
    \label{eq:Delta}
\end{equation}
Of course, by taking   $\delta^d:= \delta_K^d$, this instantiates exactly to Equation~\eqref{eq:DeltaK}, the fixpoint property of the Kantorovich behavioural distance, observed by~\cite{Desharnais02b, breugel05},  and explicitly proven in~\cite{breugelSW07}. 
One of our goals in this work was to find a similar functional that would have the $\epsilon$-distance $d^*$ as its greatest fixed point. That is, we are seeking a suitable distance $\delta^d$ on distributions. 
%that would replace $\delta_K^d$ in $\Delta_K$.
% At this point it is not clear whether we need a parameter $d$ in  $\delta^d$ but we will see that we do.

\begin{remark} \label{rem:unit-isometry}
Another property that one might consider natural and one might expect from such a distance $\delta^d$ is that the distance between the Dirac distributions on states is the same as the starting distance between them. That is, for a distance $d$ on states, one might expect
\begin{equation}
d(s,t)= \delta^d(1_{s}, 1_{t}),
    \qquad\mbox{ for $s,t\in S$,}
    \label{eq:d_and_Dirac}
\end{equation}
where we write $1_s$ for the Dirac distribution on $s\in S$. This is an interesting property of the lifting $\delta^d$ %(and not really of the wanted $d^\dagger$) 
--- it is a stronger version of the non-expansiveness of the unit of the distribution monad, showing that the unit is an isometry, and we will return to it in Section~\ref{sec:axiom} below where we show that it holds for the Lévy-Prokhorov distance, for any starting distance $d$ on states. It also holds for the (undiscounted) Kantorovich distance---this is easy to prove using the Wasserstein formulation of the Kantorovich distance,  whereas it does not hold for the total variation distance~\cite{MioSV24}. Therefore, this property is neither unique to the Lévy-Prokhorov distance that we are interested in, nor it is ``behavioural'' in the sense that we explained so far: it does not take into account the transition structure $\tau$ of the states. For a somewhat behavioural explanation, note that combining it with Equation~\eqref{eq:def_d_naive}, this condition implies that a pair of states $s', t'$ each having a single outgoing transition as a Dirac on state $s$ and $t$ respectively, would get the same distance as $s$ and $t$. In particular, such a distance does not discount the future. 
We will discuss a discounting version of the distance below.
\end{remark}

\subsection{The L\'{e}vy-Prokhorov distance on distributions} 			%------------------------------------------------------
\label{s:LP}

While examining equation~\eqref{eq:d-dagger-fixpoint} and trying to make $d^*$ fit into it as $d^\dagger$ we came up with the following distance on distributions. Only afterwards, we discovered that the distance  was actually known as the Lévy-Prokhorov lifting of a distance $d$ to distributions. In this section we introduce the distance and provide an example that illustrates it.

\begin{definition}[%(Def 8) 
Lévy-Prokhorov distance~\cite{prokhorov}] 
\label{def:dLP}
Let $\langle S,d \rangle$ be a metric space
%and $\cD S$ the set of sub-distributions %(tight) Borel probability measures on $S$
; we endow $\cD S$ with the  pseudometric
$\delta^d_{\LP}: \cD S\times \cD S \to [0,1]$, defined as 
$$ 
\delta^d_{\LP}(\mu_{0},\mu_{1}) = \inf \{\epsilon\mid \forall X\subseteq S
: \mu_{i}(X) \leq \mu_{1-i}(X^d_{\epsilon})+\epsilon, \mbox{ for } i = 0,1\}, 
$$
where $X^d_{\epsilon}=\{y\mid\exists x\in X : d(x,y)< \epsilon\}$. \\ We call $\delta^d_{\LP}$ the L\'evy-Prokhorov (LP, for short) distance.
\end{definition}

Let $d$ be a pseudometric on $S$ and $\mu_0, \mu_1 \in \mathcal{D}S$. For sake of brevity, we sometimes say that $\varepsilon \geq 0$ is \emph{admissible} for $\delta^d_{\LP}(\mu_0,\mu_1)$ if
for all $X \subseteq S$ and $i \in \{0,1\}$,
  $\mu_i(X) \leq \mu_{1-i}(X^d_\varepsilon) + \varepsilon$.
  Note that the symmetry in the definition is necessary for subdistributions. Taking $\mu_1 =1_s$, the Dirac distribution on some point $s$, and $\mu_0 =\frac{1}{2}1_s$,   we have 
  $\mu_0(X) \leq \mu_1(X^d_\varepsilon) + \varepsilon$ for all $X \subseteq S$ and $\varepsilon$, but this  does not hold with $\mu_0$ and $\mu_1$ interchanged.\\

In the following example we define simple probability measures that help illustrate the need for the extra ``$+\,\epsilon$'' in this definition and the need for the ``$\epsilon$ ball around $X$''. At first sight, $\delta^d_{\LP}$ looks very much like the total variation distance, but this ``$\epsilon$ ball around $X$'' makes it very different.

\begin{example}
\label{ex:motivate_eps_balls}
Consider the set of states $S =  \{x_\gamma \mid \gamma\in [0,1]\}$ and a distance on these states given by  $d(x_\gamma, x_\xi) = |\gamma-\xi|$, for $\gamma,\xi\in[0,1]$. In a probabilistic transition system, one could imagine that $x_\gamma$ has an $a$-loop to itself with probability $1-\gamma$, $\gamma\in [0,1]$, 
as depicted below (adding transitions out of the states to help see their differences as processes in a transition system). 
We now define a family of distributions on these states, and we picture them below as states with an outgoing transition without label (one could imagine an $a$-label). Let $\nu_{\gamma} = \gamma 1_{x_1} + (1-\gamma)1_{x_\gamma}$, for $\gamma \in [0,\frac{1}{2})$, with  $1_{x}$ the Dirac measure on $x$. The distributions  $\nu_{\gamma}$ and $\nu_{0}$ are illustrated below. For a fixed $\gamma$, these two distributions are actually non-zero on a three-state space $S' = \{x_0, x_1, x_\gamma\}$ . 
%J: I don't think nu gives more but here it is
%Finally, let $\nu = \gamma 1_{x} +\gamma 1_{x_\gamma} + (1-2\gamma)1_{y} $.
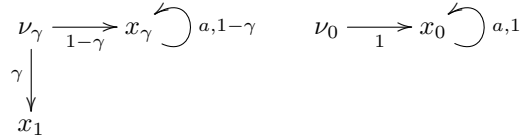
$$\xymatrix{
\nu_\gamma\ar[d]_{\gamma}\ar[r]_{1-\gamma} & x_{\gamma}\ar@(dr,ur)[]_{a,1-\gamma} 
&&  \nu_0\ar[r]_{1} & x_0\ar@(dr,ur)[]_{a,1} \\x_1%&&&\bullet
}
% \qquad
% \xymatrix@C=15mm{
% \nu \ar[r]^{1-\gamma} \ar[d]_{\gamma}\ar[rd]^{\gamma} 
% 	&x_0\ar@(dr,ur)[]_{a} \\
% 	??\ar@(dl,ul)[]^{b, ?}  &y_{\gamma}\ar@(dr,ur)[]_{a, 1-\gamma}\ar[l]^{\gamma} \\
% }
$$
  We show that  $\delta^d_{\LP}(\nu_\gamma,\nu_0) = \gamma$. Let $1-\gamma> \epsilon>\gamma$. We show that all such $\epsilon$ satisfy the inequalities in the definition of $\delta^d_{\LP}$ and any $\epsilon \le \gamma$ does not, so the distance, being the infimum over all such values, is indeed $\gamma$. In particular, we need to check that for all $X \subseteq S$, $\nu_\gamma(X) \le \nu_0(X_\epsilon^d) + \epsilon$ and $\nu_0(X) \le \nu_\gamma(X_\epsilon^d) + \epsilon$. The cases $X = S$ and $X = \emptyset$ are trivial, since the distributions are full.
  For the   sets $X$ of size two, we have 
$$\begin{array}{lllll}
    \mbox{For } X:=\{x_0,x_1\}:\quad 
        &\nu_\gamma(X) = \gamma &\!\!\leq 1 + \epsilon&=\nu_0(S)+\epsilon
        &=\nu_0(X^d_{\epsilon})+\epsilon\\
        &\nu_0(X) = 1 &\!\!\leq 1 + \epsilon&= \nu_\gamma(S)+\epsilon
        &=\nu_\gamma(X^d_{\epsilon})+\epsilon\\
    \mbox{For } X:= \{x_1,x_\gamma\}:\quad  
        &\nu_\gamma(X) = 1 &\!\!\leq  1 + \epsilon &= \nu_0(S) + \epsilon&= \nu_0(X^d_{\epsilon})+\epsilon\\ 
        &\nu_0(X) = 0 &\!\!\leq 1 + \epsilon &=\nu_\gamma(S)+\epsilon
        &=\nu_\gamma(X^d_{\epsilon})+\epsilon\\
    \mbox{For } X:= \{x_0,x_\gamma\}:\quad           &\nu_\gamma(X) =  1 - \gamma &\!\!\leq  1 + \epsilon &= \nu_0( \{x_0,x_\gamma\}) + \epsilon &= \nu_0(X^d_{\epsilon})+\epsilon \\
    &\nu_0(X) = 1  &\!\!\leq (1-\gamma) + \epsilon &=\nu_\gamma(\{x_0,x_\gamma\})+\epsilon
        &=\nu_\gamma(X^d_{\epsilon})+\epsilon\\
\end{array}$$
The last inequality is satisfied thanks to the room given by the ``$+\epsilon$''.
For sets of size one, omitting the checks where the probability is zero on the left-hand side of the inequality:
$$\begin{array}{lllll}
    \mbox{For } X:=\{x_1\}:\quad 
        &\nu_\gamma(X) = \gamma &\!\!\leq 0+\epsilon &=\nu_0(\{x_1\})+\epsilon
        &=\nu_0(X^d_{\epsilon})+\epsilon\\
    \mbox{For } X:= \{x_\gamma\}:\quad           
        &\nu_\gamma(X) = 1-\gamma &\!\!\leq 1+\epsilon &= \nu_0( \{x_\gamma,x_0\}) + \epsilon &= \nu_0(X^d_{\epsilon})+\epsilon \\
    \mbox{For } X:= \{x_0\}:\quad  
        &\nu_0(X) = 1 &\!\!\leq 1-\gamma +\epsilon &= \nu_\gamma( \{x_\gamma,x_0\}) + \epsilon&= \nu_\gamma(X^d_{\epsilon})+\epsilon. 
\end{array}$$
In the one but last inequality, it is crucial that $x_\gamma$ be within $\epsilon$ of $x_0$, which includes $x_0$  in the $\epsilon$ ball around $x_\gamma$. With $X$ in place of $X^d_\epsilon$, the inequality would not be satisfied, and not even with  $X^d_{\gamma/2}$, a set where states at distance zero are included (like bisimilar states--of which there are none here), because we would obtain: 
$$
    \nu_\gamma(\{x_\gamma\})=1-\gamma \not\leq 0+ \epsilon = \nu_0(\{x_\gamma\}^d_{\gamma/2}) +\epsilon.
$$
The last inequality also illustrates this, and also the need for the ``$+\epsilon$'', as the $\epsilon$-ball itself is not enough for the inequality to be satisfied. 
Moreover, for a pair of distributions $\nu_\gamma,\nu_\xi$ in our family one can similarly prove that $\delta^d_{\LP}(\nu_\gamma,\nu_\xi) = |\gamma- \xi|$. 
%\marginpar{Ana: shall we check?}checked

Finally, we note that for $\epsilon \le\gamma$ we have $X_\epsilon^d = X$ for any $X \subseteq S$ and hence, e.g., for $X = \{x_\gamma\}$, as already noted above we have 
$\nu_\gamma(X)=1-\gamma \not\leq 0+ \epsilon = \nu_0(X^d_\epsilon) +\epsilon$ as $\epsilon \le \gamma < \frac{1}{2}$.
\end{example}

One can notice that when viewed as processes, $\nu_\gamma$ and $\nu_0$ are the same as $s$ and $t$ of Example~\ref{ex:e-bisim}.
 Technically, it is rather $\tau_a(s)$ and $\tau_a(t)$ of course.

 Once one has seen the Lévy-Prokhorov  distance on distributions  it seems not surprising that it has some link with the  $\epsilon$-distance. However, our surprise was the other way around, when we realized that the  distance on distributions $\delta^d$ that we had to concoct to make $d^*$ a fixed point of $\Delta$ in Equation~\eqref{eq:Delta} already existed.

\subsection{The L\'evy-Prokhorov distance on LMPs}
We are now ready to define the functional on distances over LMPs that we were looking for. We expect a fixpoint of this functional to be a behavioral metric, but more importantly we expect  $d^*$ to be the greatest fixpoint of it. 

\begin{definition}
\label{Def:DeltaLP}
Let $\langle S,d \rangle$ be a pseudometric space, on which an LMP $\tau:S\to(\cD S)^\cA$ is defined. Let $s_0$ and $s_1$ be states. We define the functional $\Delta_{\LP}$ as
\begin{align*}
\Delta_{\LP}(d)(s_0,s_1)
&= 
\sup_{a\in \cA} \delta_{\LP}^d(\tau_{a}(s_{0}),\tau_{a}(s_{1}) )   
\\
&= 
\sup_{a\in \cA}\inf \{
\epsilon\mid (\forall X\subseteq S: \tau_{a}(s_{i})(X)\leq \tau_{a}(s_{1-i})(X^d_{\epsilon})+\epsilon, \quad i = 0,1)
\}.
\end{align*}
\end{definition}
Hence, for a fixed pseudometric space $\langle S,d \rangle$, $\Delta_{\LP}$ applied to $d$ gives a new pseudometric on $S$. Clearly, this can be seen as a functional on $\PMet1$, mapping $\langle S,d \rangle$ to $(S, \Delta_{\LP}(d))$.

Before proving that $d^*$ is a fixpoint, we show how 
the supremum over actions can safely be replaced by inserting a universal quantification on actions inside the set, as will be useful in Proposition~\ref{d*_isFixp}. %(The proof is in Appendix~\ref{sec:App-A}).
\begin{lemma}
\label{lemma:sup_out}
Let $\langle S,d \rangle$ be a pseudometric space, on which an LMP $\tau:S\to(\cD S)^\cA$ is defined. Let $s_0$ and $s_1$ be states. Then
\begin{align*}
\Delta_{\LP}(d)(s_0,s_1)
&= 
\inf \{
\epsilon\mid (\forall a\in \cA,\forall X\subseteq S: \tau_{a}(s_{i})(X)\leq \tau_{a}(s_{1-i})(X^d_{\epsilon})+\epsilon, \quad i = 0,1)
\}.
\end{align*}
\end{lemma}

\begin{proof}
For $a\in \cA$ and $\epsilon>0$, we define the following predicate:
$$
P(a,\epsilon):= (\forall X\subseteq S: \tau_{a}(s_{i})(X)\leq \tau_{a}(s_{1-i})(X^d_{\epsilon})+\epsilon, \quad i = 0,1).
$$
This family of predicates satisfies 
\begin{equation}
    \mbox{if }\gamma>\epsilon, P(a,\epsilon)\Rightarrow P(a,\gamma).
    \label{eq:gamma>epsilon}
\end{equation} We now prove that
$$
 \sup_{a\in \cA} \inf \{\epsilon\mid P(a,\epsilon)\}=\inf \{\epsilon\mid (\forall a\in \cA: P(a,\epsilon))\}
$$
\begin{itemize} 
    \item[$\leq$:]  The following sequence of arguments shows this inequality.
    \begin{align*}
           %\mbox{ we have } 
        \{\epsilon\mid P(a_0,\epsilon)\} 
            &\supseteq \{\epsilon\mid (\forall a\in \cA: P(a,\epsilon))\} &&\mbox{for all } a_0\in\cA, \mbox{ so}\\  
        \inf \{\epsilon\mid P(a_0,\epsilon)\} 
            &\leq \inf\{\epsilon\mid (\forall a\in \cA: P(a,\epsilon))\} &&\mbox{for all } a_0\in\cA, \mbox{ and so}\\
        \sup_{a\in \cA}\inf\{\epsilon\mid P(a,\epsilon)\} 
            &\leq \inf\{\epsilon\mid (\forall a\in \cA: P(a,\epsilon))\}. 
    \end{align*}
    \item[$\geq$:]  Let 
    $\alpha = \sup_{a\in \cA}\inf\{\epsilon\mid P(a,\epsilon)\}.$ Then for all $a\in \cA$,  $\alpha\geq \inf\{\epsilon\mid P(a,\epsilon)\}.$ Let $\gamma>\alpha$. Then we have
    \begin{align*}
        \mbox{for all } a\in \cA 
            &\quad\gamma\in \{\epsilon\mid P(a,\epsilon)\}
                &\mbox{by~\eqref{eq:gamma>epsilon} and because } \gamma> \inf\{\epsilon\mid P(a,\epsilon)\}\\
        \mbox{so  } 
            &\quad\gamma\in \cap_{a\in \cA} \{\epsilon\mid P(a,\epsilon)\}\\
        \mbox{so  } 
        &\quad\gamma\in  \{\epsilon\mid (\forall a\in \cA: P(a,\epsilon))\}\\
        \mbox{so  } 
        &\quad\gamma\geq  \inf\{\epsilon\mid (\forall a\in \cA: P(a,\epsilon))\}.
\end{align*}
This implies $\alpha\geq  \inf\{\epsilon\mid (\forall a\in \cA: P(a,\epsilon))\}$ as well. Otherwise, suppose $\alpha<\iota := \inf\{\epsilon\mid (\forall a\in \cA: P(a,\epsilon))\}$. Then $\alpha<(\alpha+\iota)/2<\iota$, so the result above with $\gamma = (\alpha+\iota)/2$ yields $(\alpha+\iota)/2\geq\iota$, a contradiction.
\end{itemize}
\end{proof}

One remarkable property of (post)fixpoints of this functional is that they define  $\epsilon$-bisimulations.

\begin{theorem}
\label{t:Reps_is_epsbisim}
Let  $\tau:S\to(\cD S)^\cA$ be an LMP with $\langle S,d \rangle$ a pseudometric space. If $d$ is a (post)fixpoint of $\Delta_{\LP}$, then for any $\epsilon>0$, $R_\epsilon := \{(s,t)\in S\times S \mid d(s,t)<\epsilon\}$, is an $\epsilon$-bisimulation.
\end{theorem}

\begin{proof}
Let $d$ be a postfixpoint (wrt~$\sqsubseteq$), i.e., let $\Delta_{\LP}(d) \leq d$.
Let $R_\epsilon=\{(s,t)\mid d(s,t)< \epsilon\}$.
We show it is an $\epsilon$-bisimulation.
Let $s_0 R_\epsilon s_1$, i.e., $d(s_0,s_1)< \epsilon$, and also $\Delta_{\LP}(d)(s_0,s_1)< \epsilon$, because $d$ is a postfixpoint.  Let $a\in\cA$; then $\delta_{\LP}^d(\tau_{a}(s_{0}),\tau_{a}(s_{1}))<\epsilon$. Because $\delta_{\LP}^d$ is an infimum, there is some $\gamma<\epsilon$ such that for $i=1,2$ and $X\subseteq S$
\begin{align*}
\tau_a(s_i)(X) 
&\leq \tau_a(s_{1-i})(X^d_\gamma) + \gamma & \mbox{by definition of $\delta_{\LP}$}\\
&\leq \tau_a(s_{1-i})(X^d_\epsilon) + \epsilon &\mbox{because }X^d_\gamma\subseteq X^d_\epsilon \mbox{ and } \gamma\leq\epsilon.\\
&= \tau_a(s_{1-i})(R_\epsilon(X)) + \epsilon & \mbox{by definition of  }R_\epsilon(X).
\end{align*}
%The same is true with 0 and 1 inverted, so 
So $R_\epsilon$ is an $\epsilon$-bisimulation, as wanted.
\end{proof}
The proof of this theorem needs that $X_\epsilon^d$ is  defined with a \emph{strict} inequality. It is necessary for the existence of $\gamma$ (and for the last equality).

This theorem is nice in itself but it also has an important corollary, that $d^*$ is greater than or equal to all fixpoints of $\Delta_{\LP}$ wrt $\sqsubseteq$.

\begin{corollary}
\label{p:d*greater}
 $d^*$ is greater than or equal to all fixpoints, i.e.,
$$\Delta_{\LP}(d) = d \quad \Longrightarrow \quad d \sqsubseteq d^*.$$
\end{corollary}

\begin{proof}%{$d\sqsubseteq d*$}  
Let $d$ be a fixpoint of $\Delta_{\LP}$ and let $d(s,t) = \epsilon$. By the previous theorem, $s\sim_\gamma t$ for all $\gamma>\epsilon$. 
Since $d^*(s,t) = \inf \{e\mid s\sim_e t\}$, we obtain $d^*(s,t)\leq \epsilon = d(s,t)$, as wanted.
\end{proof}

We prove that $\Delta_{\LP}$ is monotone, after the following auxiliary lemma. 
\begin{lemma} 
\label{lemma:d1<d2}
If $d_1\le d_2$ (equivalently, if $d_2 \sqsubseteq d_1$), we have $X^{d_2}_\epsilon\subseteq X^{d_1}_\epsilon$.  
\end{lemma}

\begin{proof} 
Assume $d_1\leq d_2$. Let $y\in X^{d_2}_\epsilon$. Then there is some $x\in X$ such that $d_2(x,y)< \epsilon$. This gives $d_1(x,y)\leq d_2(x,y)<\epsilon$. So $y\in X^{d_1}_\epsilon$.
\end{proof}

\begin{proposition}\label{prop:deltaLP_mono}
$\Delta_{\LP}$ is monotonic.
\end{proposition}

\begin{proof}{}  
Let $d_1\leq d_2$. 
Let $s_0$, $s_1$ be states, let $a\in\cA$ and define the two sets
\begin{align*}
E_1^a&= \{
\epsilon\mid (\forall X\subseteq S:  \tau_{a}(s_{i})(X)\leq \tau_{a}(s_{1-i})(X^{d_1}_{\epsilon})+\epsilon, \quad i = 0,1)
\}\\
 \mbox{and} 
\quad
E_2^a &= \{
\epsilon\mid (\forall X\subseteq S: \tau_{a}(s_{i})(X)\leq \tau_{a}(s_{1-i})(X^{d_2}_{\epsilon})+\epsilon, \quad i = 0,1)
\}.\end{align*}
Then
$E_2^a\subseteq E_1^a$. Indeed, let $\epsilon\in E_2^a$ and $X\subseteq S$. Then
$$\tau_{a}(s_{i})(X)\leq \tau_{a}(s_{1-i})(X^{d_2}_{\epsilon})+\epsilon\leq \tau_{a}(s_{1-i})(X^{d_1}_{\epsilon})+\epsilon,$$ because $X^{d_2}_\epsilon\subseteq X^{d_1}_\epsilon$ by Lemma~\ref{lemma:d1<d2}.
So $\inf E_1^a \leq \inf E_2^a$ and hence $$\Delta_{\LP}(d_1)=\sup_{a\in \cA}\inf E_1^a\leq \sup_{a\in \cA}\inf E_2^a=\Delta_{\LP}(d_2),$$ as wanted.
\end{proof}

Here is the long-awaited result, that $\epsilon$-distance is a fixpoint of $\Delta_{\LP}$.
\begin{proposition} 
\label{d*_isFixp}
$d^*$   is a fixpoint of $\Delta_{\LP}$, and hence it is the greatest fixpoint. 

\end{proposition}

\begin{proof}%[Proof of Proposition~\ref{d*_isFixp}]%{$d^*$ is a fixpoint} ------------------PROOF
Let $s_{0}$, $s_{1}\in S$. We want that $\Delta_{\LP} (d^*)(s_{0},s_{1})=d^*(s_{0},s_{1})$. 
Consider the following two sets 
\begin{align*}
E^\Delta &= \{
\epsilon\mid \forall a\in\cA: (\forall X\subseteq S:  \tau_{a}(s_{i},X)\leq \tau_{a}(s_{1-i},X^{d^*}_{\epsilon})+\epsilon, \quad i = 0,1)
\}\\
 \mbox{and} 
\quad
E^* &= \{\epsilon\mid s_{0}\sim_{\epsilon}s_{1}\}\\
&=
\{\epsilon\mid \exists R \mbox{ an $\epsilon$-bisimulation s.t. $s_{0} R s_{1}$, that is, for all $(t_{0},t_{1})\in R$, we have}\\
%\textcolor{blue}{Careful! The condition is on all pairs s,t}}\\
&\mbox{~~~~}\quad \forall a\in\cA:(\forall X\subseteq S: \tau_{a}(t_{i},X)\leq \tau_{a}(t_{1-i},R(X))+\epsilon, \quad i = 0,1)\}.
\end{align*}
Then  $\Delta (d^*)(s_{0},s_{1}) = \inf E^\Delta$ by Lemma~\ref{lemma:sup_out} and $d^*(s_{0},s_{1})=\inf E^*$. 
We prove that $\inf E^\Delta=\inf E^*$:

\begin{itemize}
%----------------------- FIRST INCLUSION
\item[$\geq$:] We prove that $E^\Delta\subseteq E^*$. Let $\epsilon\in E^{\Delta}$. Let $R_{\epsilon}= \{(x,y)\mid x\sim_{\epsilon} y\}$, that is, the biggest $\epsilon$-bisimulation. We need to prove that $s_0\sim_\epsilon s_1$. For that we prove that $R_\epsilon^+ := R_\epsilon\cup\{(s_0,s_1)\}$ is an $\epsilon$-bisimulation. Let $a\in\cA$ and $X\subseteq S$.
Then $X_{\epsilon}^{d^*}\subseteq R^+_{\epsilon}(X)$. Indeed, let $y\in X_{\epsilon}^{d^*}$ then $d^*(y,x)< \epsilon$ for some $x\in X$. So%(\textcolor{blue}{WHY? Xh, Ok, I think I got this ;-))}\textcolor{green}{
, because $d^*(y,x)= \inf \{e\mid y\sim_e x\}$, we have $y\sim_{\gamma}x$ for some $\gamma<\epsilon$, and hence $y\sim_{\epsilon}x$ as every $\gamma$-bisimulation is an $\epsilon$-bisimulation. Hence $y\in R^+_{\epsilon}(X).$ We only need to prove the bisimulation condition for the pair $(s_0,s_1)$, as for the others the property holds from $R_{\epsilon}(X)\subseteq R^+_{\epsilon}(X)$.
Then, for $i=0,1$ we have
\begin{align*}
\tau_{a}(s_{i},X)
    &\leq \tau_{a}(s_{1-i},X_{\epsilon}^{d^*})+\epsilon 
	   &&\mbox{by choice of } \epsilon\\
    &\leq \tau_{a}(s_{1-i},R^+_{\epsilon}(X))+\epsilon 
	   &&\mbox{because } X_{\epsilon}^{d^*}\subseteq R^+_{\epsilon}(X).  
\end{align*}
So $s_0\sim_\epsilon s_1$  and hence $\epsilon \in E^*$.
%----------------------- SECOND INCLUSION
\item[$\leq$:] Let $\epsilon \in E^*$. We prove that for all $\gamma>\epsilon$, we have $\gamma\in E^\Delta$. This will prove that $\inf E^\Delta \leq\epsilon$. Since $\epsilon$ is arbitrary, this will give the result: $\inf E^\Delta \leq\inf E^*$.\\
So let $\gamma>\epsilon$. 
Because $\epsilon \in E^*$, we have $s_0\sim_\epsilon s_1$, and hence 
%We prove that all $\gamma>\epsilon$ satisfy $\gamma\in E^\Delta$. Let $\gamma>\epsilon$ and l
there is some   $\epsilon$-bisimulation $R$ such that $s_{0} R s_{1}$.      We want $\gamma\in E^\Delta$, so let $a\in\cA$ and $X\subseteq S$. We want $\tau_{a}(s_{i},X)\leq \tau_{a}(s_{1-i},X^{d^*}_{\gamma})+\gamma, \quad i = 0,1$. 
%We first prove it for $X \in \Sigma_{S}$
First observe that $R(X) \subseteq X^{d^*}_{\gamma}$. Indeed, if $y\in R(X)$, then there is some $x\in X$ such that $y R x$. Because $R$ is an $\epsilon$-bisimulation,  $d^*(y,x)\leq \epsilon$, so $y\in X^{d^*}_{\gamma}$ (we cannot say $y\in X^{d^*}_{\epsilon}$ because these sets are \emph{open} balls). Then, for $ i = 0,1$:
\begin{align}
\tau_{a}(s_{i},X)
&\leq \tau_{a}(s_{1-i},R(X))+\epsilon, &\mbox{because $R$ is an $\epsilon$-bisimulation}\nonumber\\
&\leq \tau_{a}(s_{1-i},X^{d^*}_{\gamma})+\gamma, 
	&\mbox{because } R(X) \subseteq X^{d^*}_{\gamma} \mbox{and } \epsilon<\gamma.\nonumber
   % \label{eq:LPfors_i}
\end{align}
So, $\gamma \in E^\Delta$, as wanted. 
\end{itemize}
So we have proven that $d^*$ is a fixpoint of $\Delta_{\LP}$, and it is the greatest by Proposition~\ref{p:d*greater}.

\end{proof}

%\begin{remark}
%Interestingly, this proof is easier than for the $d_K$ distance, which is not proven to be a fixpoint by itself but is proven equal to the greatest one~\cite{breugel07}.

At this point we remark again that even if it seems not surprising that  the Lévy-Prokhorov  distance on distributions  has some link with the  $\epsilon$-distance, going from the basic distance on distributions $\delta_{\LP}$ to a meaningful distance on LMP is not direct (one has to find the fixed-point properties). 
%Moreover, as far as we know, this is the first time that a distance on distributions other than the Kantorovich distance is used for characterizing behavioural distances. 
%\end{remark}

Postfixpoints of $\Delta_{\LP}$ only give zero distances to bisimilar states. 

\begin{proposition}\label{prop-what-is-postfixpoint}
A distance $d$ that is a {post}fixpoint of $\Delta_{\LP}$ satisfies $d\sqsubseteq d^*$ and
$$
d(s,t)=0 \mbox{ implies that $s$ and $t$ are bisimilar}.
$$
\end{proposition}
\begin{proof} 
By the Knaster-Tarski theorem,  $d^*$ is the greatest postfixpoint, as we just proved that it is the greatest fixpoint.
%postfixpoint $d$ satisfies $d\sqsubseteq d^*$ because , being the greatest fixpoint. 
So $d\sqsubseteq d^*$, that is,  $d^*\leq d$ and hence, $d(s,t) = 0$ implies $d^*(s,t) \leq d(s,t) = 0$, and thus $d^*(s,t) = 0$. This in turn implies that $s$ and $t$ are bisimilar, by Theorem~\ref{t-kernel d}.
\end{proof} 
We note that $d^*\leq d$ says that even if $d$ coincides with $d^*$ on bisimilar states,  $d$ overestimates the distance $d^*$ of the nonbisimilar states.

\subsubsection*{Iterative characterisation of $\varepsilon$-Distance}

Let $\mathbf{0}$ be the distance of constant value 0. We will prove that $\sup_{n} \Delta_{\LP}^n(\mathbf{0})= d^*$ for a finite set $S$. In this section, we will always use the pointwise order on distance. As a consequence, when we talk of a least fixpoint, it is a greatest fixpoint wrt the order $\sqsubseteq$.

Kleene's fixpoint theorem states that on a dcpo with bottom, any Scott-continuous function has a least fixpoint obtained as a supremum of the ascending Kleene chain. 

Hence, it suffices to show that 1-bounded pseudometrics on a set $S$ form a dcpo with bottom, and that $\Delta_{\LP}$ is Scott-continuous, in order to get an iterative characterisation of the $\varepsilon$-distance $d^*$ as a supremum of the Kleene chain. 

Clearly, pseudometrics have a bottom $\bf 0$ wrt~the pointwise order. Hence, the Kleene chain is the sequence $(\Delta_{\LP}^n({\bf 0})_{n \ge 0}$ that we are interested in. 

The dcpo property is also easily seen to hold, as we consider 1-bounded pseudo metrics, and hence for any two pseudometrics $d_1$ and $d_2$ on a set $S$, there is an upper bound ${\bf 1}$, the constant pseudometric that assigns distance 1 to non-equal states and distance 0 to equal ones. 

The only thing that remains to prove is that $\Delta_{\LP}$ is Scott-continuous, which by~\cite{Markowski76} is equivalent to preserving suprema of chains. 

Recall that we call $\varepsilon \geq 0$ \emph{admissible} for $\delta^d_{\LP}(\mu_0,\mu_1)$, or just admissible for $d$, for short, if
for all $X \subseteq S$ and $i \in \{0,1\}$,
 $ \mu_i(X) \leq \mu_{1-i}(X^d_\varepsilon) + \varepsilon$.\\

\begin{lemma}%[Upward closure of admissible values]
\label{lem:upward-closure}
Let $d$ be a pseudometric on $S$ and $\mu_0, \mu_1 \in \mathcal{D}S$. Then:
\begin{enumerate}
  \item If $\varepsilon$ is admissible and $\varepsilon \leq \varepsilon'$, then
        $\varepsilon'$ is admissible.
  \item The value $\varepsilon = 1$ is admissible, so $\delta^d_{\LP}(\mu_0,\mu_1) \leq 1$.
  \item If $\delta^d_{\LP}(\mu_0,\mu_1) < \varepsilon'$, then $\varepsilon'$ is
        admissible.
\end{enumerate}
\end{lemma}

\begin{proof} ~
\begin{enumerate}
    \item Let $\varepsilon$ be admissible and $\varepsilon \leq \varepsilon'$. If
$y \in X^d_\varepsilon$, then there is $x \in X$ with $d(x,y) < \varepsilon \leq \varepsilon'$,
so $y \in X^d_{\varepsilon'}$. Hence $X^d_\varepsilon \subseteq X^d_{\varepsilon'}$, and
for all $X \subseteq S$ and $i \in \{0,1\}$,
\[
  \mu_i(X) \leq \mu_{1-i}(X^d_\varepsilon) + \varepsilon
  \leq \mu_{1-i}(X^d_{\varepsilon'}) + \varepsilon'.
\]
So $\varepsilon'$ is admissible.
    \item Since $\mu_i$ is a subdistribution, $\mu_i(X) \leq \mu_i(S) \leq 1
\leq \mu_{1-i}(X^d_1) + 1$ for all $X \subseteq S$ and $i \in \{0,1\}$.
    \item Since $\delta^d_{\LP}(\mu_0,\mu_1)$ is the infimum of the set of admissible
values and $\delta^d_{\LP}(\mu_0,\mu_1) < \varepsilon'$, there exists an admissible
$\varepsilon$ with $\varepsilon < \varepsilon'$. By $\bf 1.$, $\varepsilon'$ is admissible.
\end{enumerate}
\end{proof}

\begin{lemma}[Monotonicity wrt the starting distance]

  If $d \leq d'$, then for all $\mu, \nu \in \mathcal{D}S$,
  $\delta^d_{\LP}(\mu,\nu) \leq \delta^{d'}_{\LP}(\mu,\nu)$.
\end{lemma}

\begin{proof}
  Let $\varepsilon$ be admissible for $\delta^{d'}_{\LP}(\mu_{0},\mu_{1})$, i.e., for all
  $X \subseteq S$ and $i \in \{0,1\}$, $\mu_i(X) \leq \mu_{1-i}(X^{d'}_\varepsilon)
  + \varepsilon$. By Lemma~\ref{lemma:d1<d2}, $d \leq d'$ implies
  $X^{d'}_\varepsilon \subseteq X^{d}_\varepsilon$, hence
  $\mu_{1-i}(X^{d'}_\varepsilon) \leq \mu_{1-i}(X^{d}_\varepsilon)$, and thus
  $\mu_i(X) \leq \mu_{1-i}(X^{d}_\varepsilon) + \varepsilon$. So $\varepsilon$ is
  also admissible for $\delta^d_{\LP}(\mu,\nu)$. Since every $\varepsilon$ admissible
  for $d'$ is also admissible for $d$, we conclude
  $\delta^d_{\LP}(\mu,\nu) \leq \delta^{d'}_{\LP}(\mu,\nu)$.
\end{proof}

In the remainder of this section we will deal with several results concerning non-decreasing sequences of pseudometrics. We therefore first note the following simple fact.

\begin{lemma}[Supremum of a non-decreasing sequence of distances] Let $(d_n)_{n \ge 0}$ be a non-decreasing sequence of pseudometrics. Then $d_\infty = \sup_n d_n$ is a pseudometric, too. 
\end{lemma}

\begin{proof}
    The sequence $(d_n)_{n \geq 0}$ is non-decreasing.
%: $d_0 = 0 \leq d_1$, and if $d_n \leq d_{n+1}$ then $d_{n+1} = \Delta_{\LP}(d_n) \leq \Delta_{\LP}(d_{n+1}) = d_{n+2}$ by monotonicity of $\Delta_{\LP}$. 
Since each $d_n$ is a pseudometric, % with values in $[0,1]$
%(by Lemma~\ref{lem:upward-closure}(2), $\delta^d_{\LP} \leq 1$). Hence 
$d_\infty$ is
also a pseudometric: symmetry is immediate using the symmetry of each $d_n$
$$d_\infty(x,y) = \sup_n d_n(x,y) = \sup_n d_n(y,x) = d_\infty(y,x)$$ and the triangle inequality holds too since
\begin{eqnarray*}
    d_\infty(x,y) + d_\infty(y,z) & = & \sup_n d_n(x,y) + \sup_n d_n(y,z)\\
    & \stackrel{(*)}{\ge} & \sup_n (d_n(x,y) + d_n(y,z)) \\
    & \stackrel{(**)}{\ge} & \sup_n d_n(x,z)\\
    & = & d_\infty(x,z)
\end{eqnarray*}
where the inequality marked by $(*)$ holds since for all $n$, $\sup_n d_n(x,y) \ge d_n (x,y)$ and $\sup_n d_n(y,z) \ge d_n(y,z)$ and hence $\sup_n d_n(x,y) + \sup_n d_n(y,z) \ge d_n(x,y) + d_n(y,z)$ which then yields $(*)$, and $(**)$ holds by the triangle inequality for $d_n$ and monotonicity of suprema.
\end{proof}

\begin{lemma}[Scott-continuity of $\delta^d_{\LP}$]\label{lem:scott-continuity-lp}
  Let $(d_n)_{n \geq 0}$ be a non-decreasing sequence of pseudometrics on $S$
  converging pointwise to $d_\infty = \sup_n d_n$. Then for all
  $\mu, \nu \in \mathcal{D}S$,
  \[
    \delta^{d_n}_{\LP}(\mu, \nu) \nearrow \delta^{d_\infty}_{\LP}(\mu, \nu).
  \]
\end{lemma} 

\begin{proof}
  By the previous lemma, $d_n \leq d_\infty$ implies
  $\delta^{d_n}_{\LP}(\mu,\nu) \leq \delta^{d_\infty}_{\LP}(\mu,\nu)$ for all $n$,
  and the sequence $\delta^{d_n}_{\LP}(\mu,\nu)$ is non-decreasing. As a consequence, $\sup_n \delta^{d_n}_{\LP}(\mu,\nu) \leq \delta^{d_\infty}_{\LP}(\mu,\nu)$. It remains to
  show that $\sup_n \delta^{d_n}_{\LP}(\mu,\nu) \geq \delta^{d_\infty}_{\LP}(\mu,\nu)$.

  Assume, towards a contradiction, that $\sup_n \delta^{d_n}_{\LP}(\mu,\nu) < 
  \delta^{d_\infty}_{\LP}(\mu,\nu)$. 
Then there exists some $\gamma$  and $\beta$ such that 
  \[
    \sup_n \delta^{d_n}_{\LP}(\mu,\nu) < \gamma < \gamma + \beta <
    \delta^{d_\infty}_{\LP}(\mu,\nu).
  \]
  Since $\gamma > \sup_n \delta^{d_n}_{\LP}(\mu,\nu) $, we have
 that $\gamma$ is admissible for $d_n$ for all
  $n$. Hence for all $X \subseteq S$, $i \in \{0,1\}$, and all $n$:
  \[
    \mu_i(X) \leq \mu_{1-i}(X^{d_n}_\gamma) + \gamma.
  \]

Let $X \subseteq S$. Since $X^{d_n}_\gamma$ is non-increasing in $n$ (by Lemma~\ref{lemma:d1<d2}), it stabilizes for $n$ large
enough (as $S$ is finite), so there exists a natural number $m$ with $X^{d_m}_\gamma = \bigcap_n X^{d_n}_\gamma$. We show $\bigcap_n X^{d_n}_\gamma \subseteq
X^{d_\infty}_{\gamma+\beta}$. Let $y \in \bigcap_n X^{d_n}_\gamma$. Then for
every $n$, there exists $x_n \in X$ such that $d_n(x_n,y) < \gamma$. 
Since $d_n$ is non-decreasing in $n$, we have for all $n$:
\[ d_1(x_n,y) \le \dots \le d_n(x_n,y) < \gamma\]
in particular
\[ d_1(x_m,y) \le \dots \le d_m(x_m,y) < \gamma \]
as well as $d_k(x_m,y) < \gamma$ for all $k \ge m$, since then $X_\gamma^{d_k} = X_\gamma^{d_m}$. 
Hence, there exists $x \in X$ (namely $x = x_m$) that satisfies $d_n(x,y) <
\gamma$ for all $n$, which implies that $d_\infty(x,y) = \sup_n d_n(x,y) \leq \gamma<
\gamma + \beta$, so $y \in X^{d_\infty}_{\gamma+\beta}$ (this is where ``$+\beta$'' is needed, for the open ball). Therefore, for $n$
large enough, $X^{d_n}_\gamma \subseteq X^{d_\infty}_{\gamma+\beta}$ and
  \[
    \mu_i(X) \leq \mu_{1-i}(X^{d_n}_{\gamma}) + \gamma\leq \mu_{1-i}(X^{d_\infty}_{\gamma+\beta}) + \gamma \leq \mu_{1-i}(X^{d_\infty}_{\gamma+\beta}) + \gamma + \beta
  \]
  This holds for all $X \subseteq S$ and $i \in \{0,1\}$, so $\gamma+\beta$ is
  admissible for $\delta^{d_\infty}_{\LP}$, i.e., $\delta^{d_\infty}_{\LP}(\mu,\nu) \leq \gamma +
  \beta < \delta^{d_\infty}_{\LP}(\mu,\nu)$, a contradiction.

\end{proof}
%
%This is the claim for $S$ infinite. No proof yet.
%\begin{lemma}[Scott-continuity of $\delta^d_{\LP}$]
%\label{lem:scott-continuity-lp}
%Let $(d_n)_{n \geq 0}$ be a non-decreasing sequence of pseudometrics on $S$ converging
%pointwise to $d_\infty = \sup_n d_n$. Then for all $\mu, \nu \in \mathcal{D}S$,
%\[
%  \delta^{d_n}_{\LP}(\mu, \nu) \searrow \delta^{d_\infty}_{\LP}(\mu, \nu).
%\]
%\end{lemma}
%
%\begin{proof} 
%\end{proof}

\begin{corollary}[Scott-continuity of $\Delta_{\LP}$]
\label{cor:scott-continuity-Delta}
Let $(d_n)_{n \geq 0}$ be a non-decreasing sequence of pseudometrics on $S$ with
$d_\infty = \sup_n d_n$. Then
\[
  \Delta_{\LP}(d_n) \nearrow \Delta_{\LP}(d_\infty).
\]
\end{corollary}

\begin{proof}
%For any $s_0, s_1 \in S$ and $a \in \mathcal{A}$, $\delta^{d_n}_{\LP}(\tau_a(s_0), \tau_a(s_1)) \le \delta^{d_{n+1}}_{\LP}(\tau_a(s_0), \tau_a(s_1))$. Therefore also $$\sup_{a \in\mathcal{A}} \delta^{d_n}_{\LP}(\tau_a(s_0), \tau_a(s_1)) \le \delta^{d_{n+1}}_{\LP}(\tau_a(s_0), \tau_a(s_1))$$
%for all $a \in\mathcal{A}$ and hence $$\sup_{a \in\mathcal{A}} \delta^{d_n}_{\LP}(\tau_a(s_0), \tau_a(s_1)) \le \sup_{a \in\mathcal{A}}\delta^{d_{n+1}}_{\LP}(\tau_a(s_0), \tau_a(s_1))$$
%i.e., the sequence of suprema is non-decreasing too. In order to show that 
%\[
 % \Delta_{\LP}(d_n)(s_0, s_1)
  %= \sup_{a \in \mathcal{A}}\, \delta^{d_n}_{\LP}(\tau_a(s_0), \tau_a(s_1))
  %\nearrow \sup_{a \in \mathcal{A}}\, \delta^{d_\infty}_{\LP}(\tau_a(s_0), \tau_a(s_1))
  %= \Delta_{\LP}(d_\infty)(s_0, s_1),
%\]
%where the convergence follows %from Lemma~\ref{lem:scott-%continuity-lp} and the fact
%that a pointwise supremum of %non-decreasing sequences is %non-decreasing.
From $d_{n}\le d_{n+1}$ and monotonicity of $\Delta_{\LP}$ (Proposition~\ref{prop:deltaLP_mono}), we have
$\Delta_{\LP}(d_n) \le \Delta_{\LP}(d_{n+1})$ hence this sequence is non-decreasing at any point $(s_0,s_1)$. Its limit is then \begin{eqnarray*}\sup_n \Delta_{\LP}(d_n)(s_0,s_1) & = & 
\sup_n \sup_{a\in\mathcal{A}} \delta_{\LP}^{d_n} (\tau_a(s_0),\tau_a(s_1))\\ 
& = & \sup_{a\in\mathcal{A}}\sup_n \delta_{\LP}^{d_n} (\tau_a(s_0),\tau_a(s_1))\\ 
& = & \sup_{a\in\mathcal{A}}\delta_{\LP}^{d_\infty} (\tau_a(s_0),\tau_a(s_1))\\
& = & \Delta_{\LP}(d_\infty)(s_0,s_1).
\end{eqnarray*}
\end{proof}

This proves the main result of this section, the announced iterative characterisation of $\varepsilon$-distance. 

\begin{proposition}
\label{prop:convergence-to-lfp}
$\Delta_{\LP}^n({\bf 0}) \nearrow d^*$, i.e.,\ $\sup_{n \geq {\bf 0}} \Delta_{\LP}^n({\bf 0}) = d^*$. \qed
\end{proposition}

% \begin{proof}
% Let $d_\infty = \sup_{n \geq 0} \Delta_{\LP}^n({\bf 0})$.

% Clearly, ${\bf 0} \leq \Delta_{\LP}({\bf 0})$, and $\Delta_{\LP}$ is monotone (Proposition~\ref{prop:deltaLP_mono}), so by induction we get
% $\Delta_{\LP}^n({\bf 0}) \leq \Delta_{\LP}^{n+1}({\bf 0})$ for all $n$, i.e., the sequence $(\Delta_{\LP}^n({\bf 0}))_{n\ge 0}$ is non-decreasing.

% Since $d^*$ is a fixed point of $\Delta_{\LP}$, it is a postfixpoint (wrt~$\sqsubseteq$), so:
% $\Delta_{\LP}(d^*) \leq d^*$. By ${\bf 0} \le d^*$ and monotonicity of $\Delta_{\LP}$ we get
% $\Delta_{\LP}({\bf 0}) \le \Delta_{\LP}(d^*) = d^*$
% and further by induction,
% $\Delta_{\LP}^n({\bf 0}) \leq d^*$ for all $n$. As a consequence, $d_\infty \leq d^*$.

% For every $\varepsilon \ge 0$, by Lemma~\ref{lem:Reps-bisim}, 
% $R_\varepsilon = \{(s,t) \mid d_\infty(s,t) < \varepsilon\}$ is an $\varepsilon$-bisimulation. Let $(s,t) \in S \times S$. If $\varepsilon > d_\infty(s,t)$, then $R_\varepsilon$ is an $\varepsilon$-bisimulation containing $(s,t)$, so $s \sim_\varepsilon t$ and hence $d^*(s,t) \le \varepsilon < d_\infty(s,t).$ 
% If there is no $\varepsilon$ with $\varepsilon > d_\infty(s,t)$, then $d_\infty(s,t) = 1$ and $d^*(s,t) \le d_\infty(s,t)$ still. In any case, $d^* \le d_\infty$.

% Combining both inequalities, $d_\infty = d^*$.
% \end{proof}

Of course, if ${\bf 0}\leq d\leq d^*$, then $\Delta_{\LP}^n(d)$ converges to $d^*$. This is not true from any distance, not even for the bisimulation distance, $d_\sim$. This distance is equal to 0 on pairs of  states that are bisimilar,  it is 1 on others.

\begin{proposition}
$\Delta_{\LP}^n(d_{\sim})$  does not converge to $d^*$.
\end{proposition}

\begin{proof}
In Example~\ref{ex:01andsim} below,  $\Delta_{\LP}(\textbf{1})$ is precisely $d_\sim$ and it is a fixpoint, so iterating will not get closer to $d^*$ which is not equal to $d_\sim$. 
\end{proof}

\subsubsection*{A discounted version of $d^*$}

The Kantorovich distance has a discounted version  obtained in a simple manner by multiplying the functional by the discounted value $c\in [0,1]$. We can do exactly the same here. Note that when we speak of least fixpoints here we refer to the pointwise order $\le$ on pseudometrics.

\begin{definition}
    Let $c\in [0,1]$. The discounted functional on LMPs is defined as
    $$
    \Delta_{\LP}^c = c ~\Delta_{\LP}.
    $$
\end{definition} 

\begin{lemma}
    The operator $\Delta_{\LP}^c$ is Scott-continuous and hence has a least fixpoint $d^c$.
\end{lemma}

\begin{proof}
Monotonicity of $\Delta_{\LP}^c$ is direct from monotonicity of $\Delta_{\LP}$ and $c \ge 0$. The rest is also immediate:

Let $(d_n)_{n \geq 0}$ be a non-decreasing sequence of pseudometrics on $S$ with
$d_\infty = \sup_n d_n$. Then, by monotonicity, the sequence $\Delta_{\LP}^c(d_n)$ is non-decreasing at any point $(s_0, s_1)$ and 
\begin{eqnarray*}\sup_{n} \Delta^c_{\LP}(d_n)(s_0,s_1) & = & \sup_n c\Delta_{\LP}(d_n)(s_0,s_1)\\
& = & c \sup_n \Delta_{\LP}(d_n)(s_0,s_1)\\
& = & c\Delta_{\LP}(d_\infty)(s_0,s_1)\\
& = & \Delta^c_{\LP}(d_\infty)(s_0,s_1).
\end{eqnarray*}
From Kleene's fixpoint theorem, $\Delta_{\LP}^c$ has a least fixpoint $d_c = \sup_n (\Delta_{\LP}^c)^n({\bf 0})$.
\end{proof}

The distance $d_c$ discounts the future in the sense we announced earlier: $ d^c(1_{s}, 1_{t})=c\,d(s,t)$ for states $s$ and $t$, seeing the Dirac distributions as the behaviour of the states in the next time instant in the future. We present a detailed example demonstrating this in the next subsection. 

In \cite{TraQapl11}, a (somewhat tedious) discounted version of $\epsilon$-bisimulation is also given. We leave for future work the correspondence of $d^c$ with a distance induced by such discounted $\varepsilon$-bisimulations. %We believe that the  it coincides with 

\subsubsection*{More examples}

%this proof is good for closed and for open balls
With the next proposition, we provide an example of a postfixpoint of $\Delta_{\LP}$: the  exact bisimilarity distance. %Then $d^*\leq d_\sim$, hence $d_\sim\sqsubseteq d^*$.

\begin{proposition}\label{prop-bis-dist-fixpoint}
The exact bisimilarity distance $d_\sim$ is a postfixpoint of $\Delta_{\LP}$.
\end{proposition}
%\marginpar{proof updated}
\begin{proof}
Recall that the postfixpoint is wrt~$\sqsubseteq$, so we want to prove $\Delta_{\LP}(d_\sim)(s, t) \leq d_\sim$.
If  $s,t\in S$ are not bisimilar, then obviously we have $\Delta_{\LP}(d_\sim)(s, t) \leq 1= d_\sim(s, t) $. If they are bisimilar, we must prove that $\Delta_{\LP}(d_\sim)(s, t) = 0$.\\
% We first prove that bisimilar states  $s,t$ satisfy  $\Delta_{\LP}(d_\sim)(s, t) = 0$.
%\textcolor{brown}{
Consider $R$ to be bisimilarity. Because $R$ is a bisimulation, we have that for all  $s_0R s_1$, $a\in\cA$ and $X\subseteq S$,
$$\tau_a(s_0)(X)  
\leq  \tau_a(s_1)(R(X))
= \tau_a(s_1)(X^{d_\sim}_{\gamma}) \le \tau_a(s_1)(X^{d_\sim}_{\gamma})+\gamma, $$
for any $\gamma\in (0,1)$, because then $X^{d_\sim}_{\gamma} = R(X)$, and the same holds for $s_0$ and $s_1$ exchanged. As a consequence, $\Delta_{\LP}(d_\sim)(s_0, s_1) = 0$, as wanted.
\end{proof}

The following two lemmas are useful in order to do the computations  in the examples below. 

\begin{lemma}[Computing $\Delta_{\LP}({\bf 0})$]\label{lem:simplefactsondelta}
For  any $x, y \in S$,
%\begin{enumerate}%\addtolength{\itemsep}{-4pt}
%\item For any pseudometric $d$ on $S$, $\delta^d_{\LP}(\delta_x, \delta_y) = d(x,y)$.
%\item 
$$\Delta_{\LP}(\mathbf{0})(x,y) = \sup_{a \in \mathcal A} |\tau_a(x)(S) - \tau_a(y)(S)| .$$ %So $\mathbf{0}$ is not necessarily a fixpoint.
%\end{enumerate}
\end{lemma}

\begin{proof}
%1. Let $E= \{\varepsilon\mid \forall X\subseteq S. \delta_{x}(X)\leq\delta_y(X^d_\varepsilon) + \varepsilon\mbox{ and }\delta_{y}(X)\leq\delta_xs(X^d_\varepsilon) + \varepsilon\}$. We prove $\leq$ by proving that any $\epsilon>d(x,y)$ is in $E$. Indeed, we then have $y \in X^d_\varepsilon$
%whenever $x \in X$, and symmetrically $x \in X^d_\varepsilon$ whenever $y \in X$.
%For any $X \subseteq S$: if $x \in X$ then $\delta_x(X) = 1 \leq 1 + \varepsilon
%= \delta_y(X^d_\varepsilon) + \varepsilon$; if $x \notin X$ then $\delta_x(X) = 0$
%and the constraint holds trivially. The symmetric constraint  is identical.
%This proves the upper bound.
%\\For $\geq$, we prove that any value $\varepsilon< d(x,y)$ is not in $E$. Indeed, for such an $\epsilon$,  take $X=\{x\}$. Then $\{x\}^d_\varepsilon = \{x\}$ so we obtain
%$1 = \delta_x(\{x\}) \not\leq \delta_y(\{x\}^d_\varepsilon) + \varepsilon=\varepsilon$, %unless $\epsilon = 1$.
%This is exactly the statement of Lemma~\ref{lem:eta-nonexp} below, which we mention here in order to calculate.}
The condition $ \tau_a(x)(\emptyset)\leq\tau_a(y)(\emptyset^\mathbf{0}_\varepsilon) + \varepsilon$ is trivially true.
For a non empty $X \subseteq S$ %$X \neq \emptyset$, 
and any $\varepsilon>0$, we have $X^\mathbf{0}_\varepsilon = S$. The condition $ \tau_a(x)(X)\leq\tau_a(y)(X^\mathbf{0}_\varepsilon) + \varepsilon$ is thus equivalent to $\tau_a(x)(X)\leq\tau_a(y)(S) + \varepsilon$, for all $X \subseteq S$ 
%(there is no need to exclude $X = \emptyset$ here, as both conditions are anyway true for $X = \emptyset$), 
and all $a \in \mathcal A$. Note that this condition is further equivalent to $\tau_a(x)(S)\leq\tau_a(y)(S) + \varepsilon$ for all $a \in \mathcal A$. The one implication ($\Rightarrow$) is simply obtained for $X = S$. For the other implication ($\Leftarrow$), assume for all $a\in \mathcal A$, $\tau_a(x)(S)\leq\tau_a(y)(S) + \varepsilon$ and let $X \subseteq S$. Then $\tau_a(x)(X) \le \tau_a(x)(S) \leq\tau_a(y)(S) + \varepsilon$, as needed. The condition holds for $x$ and $y$ exchanged, too, and hence 
\begin{eqnarray*} \Delta_{\LP}^{\bf 0}(x,y) & = & \sup_{a\in \mathcal A} \delta_{\LP}^{\bf 0}(\tau_a(x), \tau_a(y))\\
& = & \sup_{a\in \mathcal A} \inf\{ \varepsilon \mid \, | \tau_a(x)(S)-\tau_a(y)(S)| \leq \varepsilon \}\\
& = & \sup_{a\in \mathcal A}| \tau_a(x)(S)-\tau_a(y)(S)|.
\end{eqnarray*}
As a side note, if $ \tau_{a}(x)(S)$ is not constant  for all $x$ and $a$, then $\Delta_{\LP}(\mathbf{0})\neq {\bf 0}$.
\end{proof} 

\begin{lemma}[Distance to a zero subdistribution] \label{lem:distance-to-zero-distr} Let $\tau \in \cD S$ and let $\hat{0}$ denote the zero subdistribution, i.e., $\hat{0}(x) = 0$ for all $x\in S$. Then $\delta_{\LP}^d(\tau, \hat{0}) = \tau(S).$
\end{lemma}

\begin{proof}
    \begin{eqnarray*}
    \delta_{\LP}^d(\tau, \hat{0}) & = & \inf\{\varepsilon \mid \forall X\subseteq S.\,\, \tau(X) \le 0 + \varepsilon, 0 \le \tau(X_\varepsilon^d) + \varepsilon\}\\
    & = & \inf\{\varepsilon\mid \forall X\subseteq S.\,\, \tau(X) \le \varepsilon\}\\
    & = & \inf \{\varepsilon \mid \tau(S) \le \varepsilon\}\\
    & = & \tau(S).
    \end{eqnarray*}
\end{proof}

%\marginpar{maybe $X \neq \emptyset$ should not be mentioned.} 
%--\marginpar{Weird, so for full distributions $d^*$ is {\bf 0} !!}
%-- No this is not weird as if you have only full distributions and no dead states, for all actions, then all states are indeed bisim 

\begin{example}\label{ex:01andsim}
Let $\mathcal A = \{a\}$, $S=\{x_0,x_1,x_2\}$ with $\tau_a(x_0)=1_{x_1},\ \tau_a(x_1)=1_{x_1},\ \tau_a(x_2)=\gamma1_{x_2}$ for some $\gamma \in (0,\frac{1}{2}]$, as depicted below.
\[
\xymatrix{
  {x_0} \ar[r]^{a,1} & {x_1} \ar@(dr,ur)[]_{a,1}
}
\qquad\qquad
\xymatrix{
  {x_2} \ar@(dr,ur)[]_{a,\gamma}
}
\]
%where $\Delta(d)(s,t)=\delta^d_{\LP}(\tau(s),\tau(t))$.
We write $d_i=\Delta_{\LP}^i({d_0})$, for $i>0$, starting from a given distance $d_0$.

\medskip
\noindent\textbf{Starting from $d_0=\mathbf{0}$ leads to $d^*$}
\begin{itemize}
  \item $d_1(x_0,x_1)=\delta^{d_0}_{\LP}(1_{x_1},1_{x_1})= 0$.
  \item $d_1(x_0,x_2)=d_1(x_1,x_2)=\delta^{d_0}_{\LP}(1_{x_1},\gamma1_{x_2})=|1_{x_1}(S)-\gamma 1_{x_2}(S)|=1-\gamma$ by Lemma~\ref{lem:simplefactsondelta}.
  \item $d_2=\Delta_{\LP}(d_1)$. Clearly, $d_2(x_0,x_1) = \delta_{\LP}^{d_1}(1_{x_1}, 1_{x_1}) = 0$. More interestingly, $d_2(x_0,x_2) = d_2(x_1,x_2) =  \delta_{\LP}^{d_1}(1_{x_1}, \gamma1_{x_2}) = 1-\gamma$  as every $\varepsilon > 1-\gamma$ is admissible: Indeed, for $\varepsilon > 1-\gamma$, if $x_1 \in A$ then $x_2 \in A_\varepsilon^{d_1}$, and if $x_2 \in A$ then $x_1 \in A_\varepsilon^{d_1}$ and all inequalities are satisfied. In contrast, if $\varepsilon < 1-\gamma$, then, e.g. for $A = X$ we get $1 = 1_{x_1}(A) > \gamma1_{x_2}(A_\varepsilon^{d_1}) + \varepsilon = \gamma + \varepsilon$. This yields $d_2 = d_1 = d^*$.
\end{itemize}

\begin{table}[]
    \centering
\begin{tabular}{cccccc}
\hline
 $(s,s')$ & $\tau_a(s)$ & $\tau_a(s')$
  & $d_1=\Delta_{\mathrm{\LP}}(\mathbf{1})$
  & $d_2=\Delta_{\mathrm{\LP}}(d_1)$ \\
\hline
$(x_0, x_1)$ & $1_{x_1}$ & $1_{x_1}$
  & $0$
  &  $0$ \\
$(x_0, x_2)$ & $1_{x_1}$ & $\gamma 1_{x_2}$
  & $1$
  &  $1$ \\
$(x_1, x_2)$ & $1_{x_1}$ & $\gamma 1_{x_2}$
  & $1$
  &  $1$ \\
\hline
\end{tabular}
    \caption{Iterating the computation of $\Delta_{\mathrm{\LP}}$ on $d_0=\mathbf{1}$ and the LMP from Example~\ref{ex:01andsim}.}
    \label{tab:example-gamma-one}
\end{table}
\noindent\textbf{Starting from $d_0=\mathbf{1}$}, the discrete distance, is not a fixpoint as it is immediately clear that $d_1(x_0,x_1) = 0$. We also see that $d_0 > d_1$. However,  $d_1 = \Delta_{\LP}(\mathbf{1})$ is a fixpoint:
\begin{itemize}
  %\item $d_1(0,1)=\delta^{d_0}_{\LP}(\delta_1,\delta_1)=0$. So $d_{0}$ is not a fixpoint and $d_0>d_1$. But next we see that $d_{1}$ is. 
  \item $d_1(x_0,x_2)=d_1(x_1,x_2)=\delta^{\mathbf{1}}_{\LP}(1_{x_1},\gamma1_{x_2})$:
    take $X=\{x_1\}$, so $1_{x_1}(X)=1$ and $\gamma1_{x_2}(X)=0$.
    For $\epsilon \le 1$,  $X^{\mathbf{1}}_\epsilon=X$, hence $\gamma1_{x_2}(X^{\bf 1}_\epsilon)=0$.
     The condition $1=1_{x_1}(X)\leq 0+\epsilon$
    forces $\epsilon\geq 1$, so $d_1(x_0,x_2)= d_1(x_1,x_2) = 1$.
  \item $d_2=\Delta_{\LP}(d_1)=d_1$ and also here $d_2(x_0,x_1) = \Delta_{\LP}(d_1)(x_0,x_1) = \delta_{\LP}^{d_1}(1_{x_1},1_{x_1}) = 0$ and
  $d_2(x_0,x_2) = d_2(x_1,x_2) = \delta_{\LP}^{d_1}(1_{x_1},\gamma1_{x_2}) = 1$ as choosing $X = \{x_1\}$ again forces that any admissible $\varepsilon$ satisfies $\varepsilon \ge 1$. 
  Note that $\Delta_{\LP}(\mathbf{1})\neq d^*$. They differ on the pairs $(x_0,x_2)$ and $(x_1,x_2)$, on which $\Delta_{\LP}(\mathbf{1})$ gives $1$, and $d^*$ gives $1-\gamma$.
  \qedhere
\end{itemize}
\end{example}
\noindent Hence, $d^*$ is not the only fixpoint, as $d^*$ and $\Delta_{\LP}(\mathbf{1})$ are both different fixpoints.
%\marginpar{do we want to put this in evidence? YES, why not?}

\begin{example}\label{ex:dscounted-iteration}
We now give an example computation of the discounted distance.
Consider the LMP:
$$
\xymatrix{
  x_1 \ar[r]^{a,1} & x_0 \ar[r]^{a,1} & \bot \\
  y_1 \ar[r]^{a,1} & y_0 \ar[ur]_{a,1/2}
}
$$
and $d_0 =\mathbf{0}$. Then, by Lemma~\ref{lem:simplefactsondelta}, and since there is a single action $a$, we have
\[
  d_1(s,s') = \Delta^c_{\mathrm{\LP}}(\mathbf{0})(s, s')
  = c \cdot \delta^{\mathbf{0}}_{\mathrm{\LP}}\!\left(\tau_a(s),\, \tau_a(s')\right)
  = c \cdot \bigl|\tau_a(s,S) - \tau_a(s',S)\bigr|.
\]
The results are shown in the fourth column of Table~\ref{table:discounted_d}. 
\begin{table}[]
    \centering
\begin{tabular}{cccccc}
\hline
 $(s,s')$ & $\tau_a(s)$ & $\tau_a(s')$
  & $d_{1 } = \Delta^c_{\mathrm{\LP}}(\mathbf{0})$
  & $\delta^{d_1}_{\mathrm{\LP}}$ & $d_{2  }=\Delta^c_{\mathrm{\LP}}(d_1)$ \\
\hline
$(x_1, y_1)$ & $1_{x_0}$ & $1_{y_0}$
  & $0$
  & $c/2$ & $c^2/2$ \\
$(x_1, x_0)$ & $1_{x_0}$ & $1_{\bot}$
  & $0$
  & $c$ & $c^2$ \\
$(y_1, x_0)$ & $1_{y_0}$ & $1_{\bot}$
  & $0$
  & $c/2$ & $c^2/2$ \\
$(y_0, x_1)$ & $\tfrac{1}{2}1_{\bot}$ &$1_{x_0}$
  & $c/2$
  & $\max(c,1/2)$ & $\max(c^2,{c}/{2})$ \\
{$(x_0, y_0)$} & $1_{\bot}$ & $\tfrac{1}{2}1_{\bot}$
  & $c/2$
  & $1/2$ & ${c/2}$ \\
{$(y_0, y_1)$}  & $\tfrac{1}{2}1_{\bot}$& $1_{y_0}$
  & $c/2$
  &$1/2
  $& ${c/2}$ \\
$(y_0, \bot)$ & $\tfrac{1}{2}1_{\bot}$ & $\hat{0}$
  & $c/2$
  & $1/2$ & ${c/2}$ \\
$(x_0, \bot)$ & $1_{\bot}$ & $\hat{0}$
  & $c$
  & $1$ & ${c}$ \\
$(x_1, \bot)$ & $1_{x_0}$ & $\hat{0}$
  & $c$
  & $1$ & ${c}$ \\
$(y_1, \bot)$ & $1_{y_0}$ & $\hat{0}$
  & $c$
  & $1$ & ${c}$ \\
\hline
\end{tabular}
   \caption{Iterating the computation of $\Delta^c_{\LP}$ on $d_0 = {\bf 0}$ and the LMP from Example~\ref{ex:dscounted-iteration}. }
    \label{table:discounted_d}.
\end{table}    
As for the computation of $\Delta^c_{\mathrm{\LP}}(d_1)$, there are two simple cases. 
For Dirac distributions, $\delta^d_{\mathrm{\LP}}(1_x, 1_y) = d(x,y)$ (see Lemma~\ref{lem:eta-nonexp}). Such are the first three lines of the table.
When some distribution $\tau$ is compared to the zero subdistribution $\hat{0}$, the LP-distance is $\tau(S)$ (see Lemma~\ref{lem:distance-to-zero-distr}).  This is the case in the last four lines of the table. There are three cases remaining that need more attention:
\begin{itemize}
\item  For $d_{2}(y_0,x_1)= c\cdot\delta_{\LP}^{d_1}(\tfrac{1}{2}1_\bot,1_{x_0})$, the relevant constraints are 
\begin{eqnarray*}
 \tfrac{1}{2}1_\bot(X) & \le & 1_{x_0}(X_\varepsilon^{d_1}) + \varepsilon,\\
 1_{x_0}(X) & \le & \tfrac{1}{2}1_\bot(X_\varepsilon^{d_1}) + \varepsilon.
\end{eqnarray*}
These are trivially satisfied for any $X$ with $\bot \not\in X, x_0 \not\in X$.
Recall that $d_{1}(\bot,x_0)= c$. 
Therefore, for $\varepsilon \le c$, taking $X = \{x_0\}$ imposes $1 \le \varepsilon$ which contradicts $\varepsilon \le c < 1$. Hence, no $\varepsilon \le c$ is admissible. 
For $\varepsilon > c$, we have: if $\bot \in X$, then $u_0 \in X_\varepsilon^{d_1}$, and if $u_0 \in X$, then $\bot \in X_\varepsilon^{d_1}$, and the constraints are (in both cases) equivalent to $1/2 \le \varepsilon  $. Hence, the infimum over all admissible values for $\varepsilon$ is $\max(c,1/2)$ and
$d_2(y_0,x_1) = \Delta^c_{\LP}(d_1)(y_0,x_1) = c \cdot \max(c,\tfrac{1}{2})$. 
\item 
For  ${d_{2}(x_0,y_0)}=c\cdot\delta^{d_1}_{\mathrm{\LP}}(1_\bot, \tfrac{1}{2}1_\bot)$, 
the only case of interest is $\bot \in X$ which holds if and only if $\bot\in X_\varepsilon^{d_1}$, so the relevant constraints become $1 \le 1/2 + \varepsilon$ and $1/2 \le 1 + \varepsilon$, which are equivalent to $\varepsilon \ge 1/2$. So, $\delta^{d_1}_{\mathrm{\LP}}(1_\bot, \tfrac{1}{2}1_\bot) = 1/2$ and ${d_{2}(x_0,y_0)}=c/2$.
\item 
Finally, for  ${d_{2}(y_0,y_1)}=c\cdot\delta^{d_1}_{\mathrm{\LP}}(\tfrac{1}{2}1_\bot,\, 1_{y_0})$, recall that $d_1(y_0,\bot) = c/2.$ This case is similar to $d_2(y_0,x_1)$.
No value $\varepsilon \le c/2 < 1/2$ is admissible, as for $X = \{y_0\}$ the constraints are equivalent to $1 \le \varepsilon$, contradicting $\varepsilon < 1/2$. 
For $\varepsilon > c/2$, in both cases of interest ($\bot \in X$ and $y_0\in X$) the constraints are equivalent to $\varepsilon \ge 1/2$. Hence, since $c/2 < 1/2$, the infimum of all admissible values equals $1/2$ and ${d_{2}(y_0,y_1)}=\Delta_{\LP}^c(d_1)(y_0,y_1) = c/2$.
\end{itemize}
\end{example}

\subsection{The L\'evy-Prokhorov lifting of the subdistribution functor to \protect{$\PMet1$}}

Using the LP distance lifting, we define a functor $\cD_{\LP}$ on $\PMet1$, the category of 1-bounded pseudometric spaces  with nonexpansive functions, as follows:
On objects $\langle X,d \rangle$, we have 
$$
\cD_{\LP}\langle X,d \rangle = (\cD X, \delta^d_{\LP})$$
and on morphisms 
$f\colon \langle X_0,d_0 \rangle \to \langle X_1,d_1 \rangle$ we set $\cD_{\LP} f = \cD f$, that is
$$\cD_{\LP} f \colon \cD_{\LP}\langle X_0,d_0 \rangle \to \cD_{\LP}\langle X_1,d_1 \rangle \textrm{ with } \varphi \mapsto \lambda y.\, \varphi(f^{-1}(\{y\})).$$ 

\begin{proposition}\label{prop:DLP_nonexp}
$\cD_{\LP}$ is a functor on $\PMet1$.
\end{proposition}
\begin{proof}
The only thing to prove is that $\cD_{\LP}f$ is nonexpansive.  Let $\varphi_0, \varphi_1 \in \cD_{\LP}\langle X_0,d_0 \rangle$. Then (we omit the symmetric inequality for simplicity)
\begin{eqnarray*}
    \delta^{d_1}_{\LP}(\cD_{\LP} f(\varphi_0), \cD_{\LP} f(\varphi_1))
    & = & \inf \{ \epsilon \mid \forall B \subseteq X_1 : \varphi_0(f^{-1}(B)) \le \varphi_1(f^{-1}(B_\epsilon^{d_1})) + \epsilon\}\\
    & \stackrel{(*)}{\le} & \inf \{ \epsilon \mid \forall B \subseteq X_1 : \varphi_0(f^{-1}(B)) \le \varphi_1((f^{-1}B)_\epsilon^{d_0}) + \epsilon\}\\
    & = & \inf \{ \epsilon \mid \forall A \in \{f^{-1}(B)\mid B\subseteq X_1\}: \varphi_0(A) \le \varphi_1(A_\epsilon^{d_0}) + \epsilon\}\\
    & \stackrel{(**)}{\le} & \inf \{ \epsilon \mid \forall A \subseteq X_0 : \varphi_0(A) \le \varphi_1(A_\epsilon^{d_0}) + \epsilon\}\\
    & = & \delta^{d_0}_{\LP}(\varphi_0, \varphi_1)
\end{eqnarray*}
where the marked inequalities are justified below. For the first one, marked with $(*)$, let 
$$\mbox{~~~~~~}
V = \{ \epsilon \mid \forall B \subseteq X_1 : \varphi_0(f^{-1}(B)) \le \varphi_1((f^{-1}B)_\epsilon^{d_0}) + \epsilon\} \mbox{ and} $$  
$$W = \{ \epsilon \mid \forall B \subseteq X_1 : \varphi_0(f^{-1}(B)) \le \varphi_1(f^{-1}(B_\epsilon^{d_1})) + \epsilon\}.$$
We show that $V \subseteq W$. Let $\epsilon \in V$. For every $B \subseteq X_1$,  let us first show that
$ f^{-1}(B)_\epsilon^{d_0} \subseteq f^{-1}(B_\epsilon^{d_1}) $. Indeed, take $x\in f^{-1}(B)_\epsilon^{d_0} $ and let $a \in f^{-1}(B)$ be such that $d_0(x,a) < \epsilon$. Then by nonexpansivity of $f$, $$d_1(f(x),f(a)) \le d_0(x,a) < \epsilon$$
and so $f(x) \in B_\epsilon^{d_1}$, because $f(a) \in B$. So $x \in f^{-1}(B_\epsilon^{d_1})$. Now combining this with the inequality condition in $V$ we get 
$$\varphi_0(f^{-1}(B)) \le \varphi_1(f^{-1}(B)_\epsilon^{d_0}) + \epsilon \le \varphi_1(f^{-1}(B_\epsilon^{d_1})) + \epsilon$$
and so $\epsilon \in W$.
\\
For the inequality marked with (**) , observe that if $\epsilon$ satisfies the inequality for all $A\subseteq X_0$ then it does for all $A\in \{f^{-1}(B)\mid B\subseteq X_1\}$. So the infimum is taken on a possibly bigger set on the left-hand side of the inequality, which concludes the argument.
\end{proof}

Just like for the Kantorovich-lifting of the distribution functor, we can prove that the L\'evy-Prokhorov lifting of the distribution functor is locally nonexpansive, in the next proposition. For this reason, note that given pseudometric spaces $\langle X,d_X\rangle$ and $\langle Y,d_Y\rangle$ in $\PMet1$, the hom-set $X\to Y$ of all nonexpansive maps from $X$ to $Y$ carries a metric defined by
$$d_{X\to Y}(f_1,f_2) = \sup_{x \in X} d_Y(f_1(x),f_2(x)).$$

\begin{proposition}%[Prop. 17]
\label{prop:MLP_locnonexp}
The functor $\cD_{\LP}$ is locally nonexpansive, that is, for $f_{1}, f_{2}\in X\to Y$
$$\delta_{\cD_{\LP}X\to \cD_{\LP}Y}(\cD_{\LP}f_{1},\cD_{\LP}f_{2}) \leq d_{X\to Y}(f_{1},f_{2})
$$
with $\delta_{\cD_{\LP}X\to \cD_{\LP}Y}$ being the metric on the hom-set $\cD_{\LP}\langle X,d_X\rangle\to \cD_{\LP}\langle Y,d_Y\rangle$.
\end{proposition}
\begin{proof}
Let $f_1, f_2\in X \to Y$. We have that 
$$\delta_{\cD_{\LP}X\to \cD_{\LP}Y}(\cD_{\LP}f_{1},\cD_{\LP}f_{2}) = \sup_{\varphi \in \cD X}\delta_{\LP}^{d_Y}(\cD_{\LP}f_{1}(\varphi),\cD_{\LP}f_{2}(\varphi)).$$
Let $\varphi\in \cD X$ and recall that $\cD_{\LP}f_{i}(\varphi) = \varphi(f_i^{-1}(\cdot))$. We will show that 
$$\delta_{\LP}^{d_Y}(\varphi(f_1^{-1}(\cdot)), \varphi(f_2^{-1}(\cdot))) \le d_{X\to Y}(f_{1},f_{2}).$$
Let $\alpha = d_{X\to Y}(f_{1},f_{2})= \sup_{x \in X} d_Y(f_1 (x),f_2(x))$. Let $B\subseteq Y$.
We  show $f_1^{-1}(B)\subseteq f_2^{-1}(B^{d_Y}_\alpha)$
\begin{align*}
f_1^{-1}(B)
&=\{ x \mid \exists y\in B : y = f_1(x)\}
\\
&= \{ x \mid \exists y\in B : y=f_1(x) \mbox{ and } d_Y(y,f_2(x))\leq\alpha\}
&\mbox{since } d_Y(f_1 (x),f_2(x))\leq\alpha
\\
&\subseteq  \{ x \mid \exists y\in B : d_Y(y,f_2(x))\leq\alpha\}
\\
&=  \{ x \mid f_2(x)\in B^{d_Y}_\alpha\}
\\
&= f_2^{-1}(B^{d_Y}_\alpha).
%\label{}
\end{align*}

This implies that $\alpha$ satisfies $(\forall B\subseteq Y : \varphi(f_1^{-1}(B))\leq \varphi(f_2^{-1}(B^{d_Y}_\alpha))+\alpha)$, and similarly with 1 and 2 inverted.
\begin{align*}
\delta_{\LP}^{d_Y}(\varphi(f_1^{-1}(\cdot)),\varphi(f_2^{-1}(\cdot)))
&= \inf \{\epsilon\mid \forall B\subseteq Y : \varphi(f_1^{-1}(B))\leq \varphi(f_2^{-1}(B^{d_Y}_\epsilon))+\epsilon\}
\\
&\leq \alpha \,\, = \,\, d_{X\to Y}(f_{1},f_{2}),
%\label{}
\end{align*}
as wanted.

\end{proof}

For the following property, recall that a map $f\colon X\to Y$ for (pseudo)metric spaces $\langle  X, d_X\rangle$ and $\langle Y, d_Y\rangle$ is an isometry if and only if for all $x,y \in X$, $d_Y(f(x),f(y))=d_X(x,y)$.  

\begin{proposition} \label{prop:isometries}
The functor $\cD_{\LP}$ preserves isometries, i.e., if $f\colon X \to Y$ is an isometry for (pseudo)metric spaces $\langle  X, d_X\rangle$ and $\langle Y, d_Y\rangle$, then so is $\cD_{\LP}f$.
\end{proposition}

\begin{proof}
Let $f\colon X \to Y$ be an isometry from the pseudometric space $(X,d_X)$ to the pseudometric space $(Y,d_Y)$. Consider $\mathcal{D}_{\LP}f = \mathcal{D}f\colon \mathcal{D}X \to \mathcal{D}Y$. Recall that for $\varphi \in \mathcal{D}X$, $\mathcal{D}f(\varphi) = \varphi(f^{-1}(\cdot))$.
Now, for $\varphi, \psi \in \mathcal{D}X$ (again ignoring the second inequality for simplicity, which is safe to do)
\begin{eqnarray*}\delta_{\LP}^{d_Y}(\mathcal{D}_{\LP}f(\varphi), \mathcal{D}_{\LP}f(\psi))
& = & \delta_{\LP}^{d_Y}(\varphi(f^{-1}(\cdot)), \psi(f^{-1}(\cdot)))\\
& = & \inf\{ \varepsilon \mid \forall B \subseteq Y.\,\, \varphi(f^{-1}(B)) \le \psi(f^{-1}(B_\varepsilon^{d_Y})) + \varepsilon\}\\
& \le & \inf\{ \varepsilon \mid \forall A \subseteq X.\,\, \varphi(A) \le \psi(A_\varepsilon^{d_X}) + \varepsilon\}.
\end{eqnarray*}
Where the last inequality holds as in the proof of Proposition~\ref{prop:DLP_nonexp}, for arbitrary function $f\colon X\to Y$. We will still prove the opposite inequality for an isometry $f\colon X\to Y$. For this, it suffices to prove that 
$$\{\varepsilon \mid \forall B\subseteq Y.\,\, \varphi(f^{-1}(B)) \le \psi(f^{-1}(B_\varepsilon^{d_Y})) + \varepsilon\} \subseteq \{ \varepsilon \mid \forall A \subseteq X.\,\, \varphi(A) \le \psi(A_\varepsilon^{d_X}) + \varepsilon\}$$
as the infimum of a larger set is smaller. So, let $\varepsilon$ be such that $\forall B\subseteq Y.\,\, \varphi(f^{-1}(B)) \le \psi(f^{-1}(B_\varepsilon^{d_Y})) + \varepsilon$. 

Before we show that $\forall A \subseteq X.\,\, \varphi(A) \le \psi(A_\varepsilon^{d_X}) + \varepsilon$, we will show that for arbitrary $A \subseteq X$ we have that $A_\varepsilon^{d_X} = f^{-1}(f(A_\varepsilon^{d_X}))$, even though $f$ is not necessarily injective in pseudometric spaces (as it necessarily is in metric spaces). Clearly, always,  $A_\varepsilon^{d_X} \subseteq f^{-1}(f(A_\varepsilon^{d_X}))$.  Let $a \in f^{-1}(f(A_\varepsilon^{d_X}))$. Then $f(a) \in f(A_\varepsilon^{d_X})$ and hence there exists $a' \in A_\varepsilon^{d_X}$ with $f(a) = f(a')$. Hence $d_Y(f(a),f(a')) = 0$. This, since $f$ is an isometry, means $d_X(a,a') = 0$. Further, from $a' \in A_\varepsilon^{d_X}$, there exists $a'' \in A$ such that $d_X(a',a'') < \varepsilon$. Together, by the triangle inequality, we get
$$d_X(a,a'') \le d_X(a,a') + d_X(a', a'') < 0 + \varepsilon = \varepsilon$$ 
and hence $a \in A_\varepsilon^{d_X}$. 

Finally, for $A \subseteq X$ we have $A \subseteq f^{-1}(f(A))$ and hence 
$$\varphi(A) \le \varphi(f^{-1}(f(A)) \le \psi(f^{-1}(f(A)_\varepsilon^{d_Y})) + \varepsilon = \psi(A_\varepsilon^{d_X}) + \varepsilon$$
where in the second inequality we use the assumption on $\varepsilon$ for the set $B = f(A)$. This completes the proof. 
\end{proof}

As a consequence, from the results of~\cite{DBLP:conf/fossacs/GoncharovHNSW23}, we  obtain that the lifting functor $\cD_{\LP}$ \emph{is Kantorovich}, i.e., it \emph{has codensity lifting} which means that there is a set of predicate liftings of type $\cD[0,1] \to [0,1]$ such that codensity lifting gives rise to $\cD_{\LP}$. This has already been remarked in~\cite[Example 5.5.2]{DBLP:conf/fossacs/GoncharovHNSW23}. However, the exact set of predicate liftings needed to give rise to the LP-lifting is not provided there -- it has been provided in recent work~\cite{DBLP:conf/fossacs/WildSMKF26}.   

\subsubsection*{The discounting functor}

Also the discounting L\'evy-Prokhorov distance on distributions provides a functor on $\PMet1$, as follows.
On objects $\langle X,d \rangle$, we have 
$$
\cD_{\LP}^c\langle X,d \rangle = \langle\cD X, c\delta^d_{\LP}\rangle$$
and on morphisms 
$f\colon \langle X_0,d_0 \rangle \to \langle X_1,d_1 \rangle$ we set $\cD_{\LP}^c f = \cD f$.

The following properties are then almost immediate consequences of the corresponding properties of the functor $\cD_{\LP}$.

\begin{proposition}
$\cD_{\LP}^c$ is a functor on $\PMet1$.    
\end{proposition}

\begin{proof}
    Since $\delta_{\LP}^d$ is a 1-bounded pseudometric, $c\delta_{\LP}^d$ is too.

    It remains to see that $\cD_{\LP}^cf$ is nonexpansive for $f\colon \langle X_0, d_0\rangle \to \langle X_1, d_1 \rangle$ nonexpansive. Let $\varphi_0, \varphi_1 \in \cD_{LP}^cX_0 = \cD X_0$. Using Proposition~\ref{prop:DLP_nonexp}, and the fact that $c \ge 0$, we have
    $$ c\delta_{\LP}^{d_1}(\cD f(\varphi_0), \cD f(\varphi_1)) \le c\delta_{\LP}^{d_0}(\varphi_0,\varphi_1).$$
\end{proof}

\begin{proposition}
    The functor $\cD_{\LP}^c$ is locally $c$-Lipschitz. 
\end{proposition}

\begin{proof}
    Consider $f_1, f_2 \colon \langle X, d_X\rangle \to \langle Y, d_Y\rangle$. We have that
\begin{eqnarray*}
    \delta_{\cD_{\LP}^cX\to \cD_{\LP}^cY}(\cD_{\LP}^cf_{1},\cD_{\LP}^cf_{2}) & = & \sup_{\varphi \in \cD X} c\delta_{\LP}^{d_Y}(\cD f_1(\varphi), \cD f_2(\varphi))\\
    & = & c \cdot \sup_{\varphi \in \cD X} \delta_{\LP}^{d_Y}(\cD f_1(\varphi), \cD f_2(\varphi))\\
    & = & c\cdot \delta_{\cD_{\LP}X\to \cD_{\LP}Y}(\cD_{\LP}f_{1},\cD_{\LP}f_{2})\\
    & \le & c \cdot d_{X \to Y}(f_1,f_2).
\end{eqnarray*}
where the inequality uses local non-expansivity of $\cD_{\LP}$, Propositon~\ref{prop:MLP_locnonexp},  and $c\ge 0$.
\end{proof}

\begin{proposition}
    The functor $\cD_{LP}^c$ preserves isometries.
\end{proposition}

\begin{proof}
    Let $f\colon \langle X,d_X\rangle \to \langle Y,d_Y\rangle$ be an isometry, i.e., $d_Y(f(x),f(y)) = d_X(x,y)$. Then $\cD_{\LP}f$ is an isometry, by Proposition~\ref{prop:isometries}, and hence
    $$ c\cdot \delta_{\LP}^{d_Y}(\cD f(\varphi_0), \cD f(\varphi_1)) = c\cdot \delta_{LP}^{d_X}(\varphi_0,\varphi_1).$$
\end{proof}

\subsection{The lifted functor $\cD_{\LP}$ is not a monad lifting of $\cD$}\label{sec:axiom}

Behavioural distances have been axiomatized within the line of work on quantitative equational theories~\cite{MardarePP16}. It was therefore a natural question for us whether the $\epsilon$-distance, or the Lévy-Prokhorov distance itself can be given a quantitative axiomatization. This is what we briefly investigate in this section. 
%Moreover, related to it and to Remark~\ref{rem:unit-isometry}, we reveal a simple but interesting formulation of the Kantorovich distance in Remark~\ref{rem:unit-again} below. 

Mio et al.~\cite[Lemma 7.2, Theorem 7.7(2)]{MioSV24} have shown that: 
\begin{itemize}
\item If a monad on metric spaces is axiomatizable, then (just by being a monad on metric spaces) it has nonexpansive unit and multiplication.
\item A monad on metric spaces that is a lifting of a monad on $\Sets$ is axiomatizable with a quantitative theory. Being a lifting means that it acts on objects and arrows in metric spaces in the same way as it does in $\Sets$ and that the unit and multiplication are nonexpansive. 
\end{itemize}
 The Kantorovich lifting of $\cD$ has been axiomatized, and the total variation distance has been shown non-axiomatizable (as the unit is not nonexpansive). In the case of $\cD_{\LP}$ we can see that the unit is nonexpansive, but the multiplication is not, and hence $\cD_{\LP}$ is not a lifting of the monad $\cD$ to $\PMet1$ and the L\'evy-Prokhorov distance on distributions is not axiomatizable with a quantitative theory, at least not with the standard multiplication.

Recall the definitions of $\eta$ and $\mu$ for the subdistribution monad $\cD$:  
$$\eta_X(x) = 1_x, \text{ the Dirac distribution at } x; \text{ and } \mu_X(\Phi)(x)  = \sum_\varphi \Phi(\varphi)\cdot \varphi(x).$$

\begin{lemma}\label{lem:eta-nonexp}
    The unit $\eta$ of $\mathcal{D}$ is nonexpansive with respect to the Lévy-Prokhorov distance. Moreover, it is an isometry, i.e., $\delta_{\LP}^d(1_x, 1_y) = d(x,y)$.
\end{lemma}

\begin{proof}
    We have 
    \begin{align*}
    \delta^d_{\LP}(\eta_X(x), \eta_X(y))& = \delta^d_{\LP}(1_x, 1_y)\\
    &= \inf\{\varepsilon \mid \forall A \subseteq X : 1_x(A) \le 1_y(A^d_\varepsilon) + \varepsilon, 1_y(A) \le 1_x(A^d_\varepsilon) + \varepsilon\}.
    \end{align*}
    Let $A$ be a subset of $X$. We distinguish four cases in which we consider the relevant inequalities:
    \begin{align}
        1_x(A) &\le 1_y(A^d_\varepsilon) + \varepsilon \label{LP-first-ineq}\\
       \mbox{ and~~  } 1_y(A) &\le 1_x(A^d_\varepsilon) + \varepsilon
        \label{LP-second-ineq}
    \end{align}
   \begin{enumerate}
       \item $x \notin A, y \notin A$: Then $1_x(A) = 0$ and $1_y(A) = 0$, and any $\varepsilon \ge 0$ satisfies both~(\ref{LP-first-ineq}) and~(\ref{LP-second-ineq}).  
       \item $x\notin A, y \in A$: Here $1_x(A) = 0$ and $1_y(A) = 1$, and ~(\ref{LP-first-ineq}) holds always, but~(\ref{LP-second-ineq}) holds for $\varepsilon > d(x,y)$ as then $x \in A^d_\varepsilon$ and hence $1_x(A^d_\varepsilon) = 1$, and might not hold for $\varepsilon \le d(x,y)$, e.g. when $A = \{y\}$ and $\varepsilon < d(x,y) \le 1$, as then $1_x(A^d_\varepsilon) = 0$. 
       \item $x \in A, y \notin A$: This case is dual to the second case. Here~(\ref{LP-second-ineq}) always holds, but~(\ref{LP-first-ineq}) holds for $\varepsilon > d(x,y)$ and might not hold for $\varepsilon \le d(x,y)$ as, e.g., for $A = \{x\}$ we have $y \in A^d_\varepsilon$ if and only if  $\varepsilon > d(x,y)$.
       \item $x \in A, y \in A$: Then $1_x(A) = 1$ and $1_y(A) = 1$, as well as $1_x(A^d_\varepsilon) = 1_y(A^d_\varepsilon) = 1$, and again, as in the first case, any $\varepsilon \ge 0$ satisfies both~(\ref{LP-first-ineq}) and~(\ref{LP-second-ineq}). 
   \end{enumerate}
   Hence, both inequalities are satisfied for all $A$ iff $\varepsilon > d(x,y)$, and therefore $\delta^d_{\LP}(1_x, 1_y) = d(x,y)$.
\end{proof}

The unit of the distribution monad is also an isometry with respect to the Kantorovich distance (without discount), as is easy to prove using the duality with the Wasserstein distance.

However, the multiplication $\mu$ of the monad $\cD$ is not nonexpansive, as the following example shows. 

\begin{example}\label{ex:mu-not-nonexp}
    Consider $a$, $b$ and $c$ in the figure below as distributions over $\{\bot,\bullet\}$.
    %, hence subdistributions over $\{d\}$, 
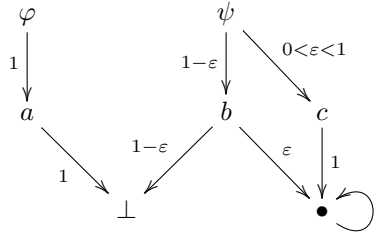
$$\xymatrix{
\varphi \ar[d]_{1}
	&&\psi\ar[d]_{{{1-\epsilon}}}\ar[dr]^{0<\epsilon<1}\\	
a \ar[dr]_{1} &&b\ar[rd]^{\epsilon} \ar[ld]_{1-\epsilon}&c\ar[d]^1\\
	&\bot  	&&	\bullet\ar@(dr,ur)[]_{} 
}$$
The ambient distance is $d(\bot,\bullet) = 1$,  we omit it from the notation: we write $\delta_{\LP}$ for $\delta_{\LP}^d$. % and $A_\epsilon$ for $A_\epsilon^d$. 
Consider  the distributions $\varphi:= 1_a$ and $\psi:=(1-\epsilon)1_b + \epsilon 1_c$   over these distributions, as pictured. We collect several useful facts: 
\begin{itemize}
\item $\delta_{\LP}(a,b)=\epsilon$ and 
$\delta_{\LP}^{\delta_{\LP}}(\varphi,\psi)=\epsilon$. 
\item $\mu \varphi = a$
 and $\mu \psi = (1-\epsilon)^2 \bot+ (1-\epsilon)\epsilon \bullet+\epsilon \bullet= (1-\epsilon)^2 \bot+ (2-\epsilon)\epsilon \bullet$. 
\item $\{\bot\}_{\gamma}=\{\bot\}$ for any $\gamma \le \epsilon$ and hence
$\mu \varphi(\{\bot\})=1$ and $\mu \psi (\{\bot\}_{\gamma})=\mu \psi (\{\bot\})= (1-\epsilon)^2$.
\end{itemize}

\noindent For $\gamma < \epsilon(2-\epsilon)$, we have $1\not\leq  (1-\epsilon)^2 +\gamma $ which therefore yields
$$\mu \varphi(\{\bot\})=1\not\leq  (1-\epsilon)^2 +\gamma = \mu \psi (\{\bot\}_{\gamma}) +\gamma .$$
\noindent This shows that $\delta_{\LP}(\mu \varphi,\mu \psi) \ge \epsilon(2-\epsilon) > \epsilon = \delta_{\LP}^{\delta_{\LP}}(\varphi,\psi)$ and hence $\mu$ is not nonexpansive. 

\end{example}

%\color{blue}
This result is not surprising, as the binding operator $\mu$ really is a multiplication, and so it accords well with the Kantorovich metric, which does multiply the probabilities of distributions in play. We do not know yet if another operator could help express the functor $\cD_{\LP}$ as a monad lifting.

\section{$\epsilon$-(Bi)simulations, coalgebraically}\label{sec:coalg-eps-bisim}

The Kantorovich behavioural distance has a coalgebraic characterization via a final coalgebra semantics. Whether the $\epsilon$-distance is obtained by finality is still open, but the results on the Kantorovich distance do not seem to directly apply here. However, as the distance is defined via $\epsilon$-bisimulations and bisimilarity, it is natural to see whether a coalgebraic semantics arising from generalizing coalgebraic bisimulations and bisimilarity is in place. The answer to this question is positive and we present the necessary observations in this section. An abstract (in this case coalgebraic) characterization allows for, on the theoretical side, deeper and clearer understanding of the notion under study, and, on the practical side, generalizations. We could use this generality to define notions of $\epsilon$-bisimulations for other types of (probabilistic) systems. 

In a nutshell, in this section we show that $\epsilon$-simulation and $\epsilon$-bisimulation have a span-diagram characterization, in a way similar to coalgebraic simulation~\cite{HughesJacobs04,Hasuo06,Hasuo10} and Aczel-Mendler coalgebraic bisimulation~\cite{AM89:ctcs,Jacobs2016-book, Rut00:tcs}. The development is using a notion of $\epsilon$-coupling defined in~\cite{TraQapl11}. For simplicity, we ignore the labels in this section, i.e., we assume there is a single label.

To start with, we recall the basic notions related to coalgebraic (bi)simulation, formulated for the functor $\mathcal D$ that we are mainly interested in in this work. Recall the definition of the subdistribution functor $\cD$ on $\Sets$: It maps a set $X$ to $\cD X$, the set of discrete subdistributions on $X$
and on arrows $f\colon X \to Y$, $\cD f\colon \cD X \to \cD Y$ is defined by
\[
\cD f(\varphi)(y) = \sum_{x \in f^{-1}(\{y\})} \varphi(x) = \varphi(f^{-1}(\{y\}).
\]

Coalgebras of the functor $\mathcal D$ are Markov chains, formally they are pairs $(X,c)$ of a carrier set $X$ and transition map  $c\colon X \to \mathcal D X$. We often just refer to the coalgebra by the transition map $c$. Given two such coalgebras, $c\colon X \to \mathcal D X$ and $d\colon Y \to \mathcal D Y$ a coalgebra homomorphism from $(X,c)$ to $(Y,d)$ is a map $h\colon X \to Y$ making the following diagram commute:
\[
\xymatrix@R=1.5em{
X
\ar[r]^{h}
\ar[d]_{c}
&
Y
\ar[d]^{d}
\\
\cD X
\ar[r]^{\cD h}
&
\cD Y}
\]

\begin{definition}\label{def:coalg-bis} Given two coalgebras $c\colon X \to \mathcal D X$ and $d\colon Y \to \mathcal D Y$, a coalgebraic bisimulation is a relation $R \subseteq X\times Y$ such that there exists a coalgebra structure $b\colon R \to \mathcal D R$ making the two projections $\pi_1\colon R \to X$ and $\pi_2\colon R \to Y$ coalgebra homomorphisms, i.e., making the following span diagram commute: 
\[
\xymatrix@R=1.5em{
X
\ar[d]_{c}
&
R
\ar[l]_{\pi_1}
\ar[r]^{\pi_2}
\ar[d]_{b}
&
Y
\ar[d]^{d}
\\
\cD X
&
\cD R
\ar[l]_{\cD\pi_1}
\ar[r]^{\cD\pi_2}
&
\cD Y}
\]
\end{definition}

It is well known, see~\cite{deVink99,Sokolova05} that this notion of bisimulation corresponds to the standard notion of probabilistic bisimulation due to Larsen and Skou~\cite{Larsen91}. For the case of continuous distributions, taking a cospan instead of a span helps avoiding the need for analytic spaces~\cite{Danoscmcs}. Back to the discrete case, one way to define the coalgebra structure $b$ is using the notion of a coupling: If for all $(x,y) \in R$, there is a coupling $\beta$ for $\mu = c(x)$ and $\nu = d(y)$, then setting $b(x,y) = \beta$ provides the needed transition structure. Recall that $\beta \in \cD R$ is a \emph{coupling} of $\mu \in \cD X$ and $\nu \in \cD Y$ if its marginals are $\mu$ and $\nu$, respectively, that is
$$\sum_{y\in Y} \beta(x,y) = \mu(x), \quad \sum_{x\in X} \beta(x,y) = \nu(y) .$$
Note that this definition of a coupling, fitting the definition of coalgebraic bisimulation for the subprobability distribution functor, ensures that a coupling of two subdistributions exists only if they have the same total mass. One could also define couplings by adding a dummy element and hence viewing a subdistribution as a full distribution, which provides a more general notion. However, these details are not relevant for what we are really interested in, which are $\epsilon$-bisimulations. 

We are going to use a relaxed notion of $\epsilon$-coupling~\cite{TraQapl11} to define $\epsilon$-(bi)simulation in analogy with the coalgebraic definition above. Before we recall those, let us first mention coalgebraic simulations, due to~\cite{HughesJacobs04,Hasuo06,Hasuo10} with a coalgebraic formulation for LMPs already in~\cite{Desharnais99c}.

For this, note that an ordered functor $F$ is a functor with order on each object $FX$. The original definition in~\cite{HughesJacobs04} uses preorders as minimal requirement, and also implies that for any map $f\colon X \to Y$, $Ff$ preserves the order, i.e., is monotone -- a condition that is not used in our results. The subdistribution functor $\cD$ is ordered, e.g., by pointwise order. Note that this order becomes trivial in case of the distribution functor which is the reason why we work with subdistributions here. Coalgebraic simulation can be defined for ordered functors, here we only recall the definition in the special case of the subdistribution functor $\cD$. All notions  involve \emph{lax} and \emph{oplax morphism}. 
\begin{definition}\label{def:lax-hom}
  A \emph{lax homomorphism} from $(X,c)$ to $(Y,d)$ is a morphism $l\colon X \to Y$ with the property that $d\circ l \sqsubseteq \cD l \circ c$, i.e., it makes the left lax diagram below commute. An \emph{oplax homomorphism} from $(X,c)$ to $(Y,d)$ is a lax homomorphism for the dual/opposite order, that is a morphism $o\colon X \to Y$ that makes the right diagram below commute:
  \[
{\xymatrix@R=1.5em{
X
\ar[r]^{l}
\ar[d]_{c}
\ar@{}[dr]|{\sqsupseteq}
&
Y
\ar[d]^{d}
&&
X
\ar[r]^{o}
\ar[d]_{c}
\ar@{}[dr]|{\sqsubseteq}
&
Y
\ar[d]^{d}
\\
\cD X
\ar[r]_{\cD l}
&
\cD Y
&&
\cD X
\ar[r]_{\cD o}
&
\cD Y}
}
\]
The \emph{order on subdistributions} is defined pointwise: For a set $A$ and $\varphi, \psi \in \cD A$, we have $\varphi \sqsubseteq \psi$ iff for all $a \in A$, $\varphi(a) \le \psi(a)$.  
\end{definition}

Note that, equivalently, $\varphi \sqsubseteq \psi$ iff for all $B \subseteq A$, $\varphi(B) \le \psi(B)$, where the right-to-left implication follows from instantiating on singleton subsets and the left-to-right implication follows from the additivity of distributions. 

Initially, generic coalgebraic simulations have been studied for ordered functors, by Hughes and Jacobs~\cite{HughesJacobs04}, using a notion of lax relation lifting. On the other hand, Hasuo~\cite{Hasuo06,Hasuo10} discovered (generalizations of) forward, backward, and hybrid simulations in Kleisli categories and moreover showed in~\cite{Hasuo10} that Hughes-Jacobs simulations are a special case. While Kleisli categories are appealing for soundness results, and monads suitable for traces come equipped with an order, we may safely ignore all Kleisli aspects here. We now recall the following notions of simulations~\cite{Hasuo06}, formulated for general categories of coalgebras (as long as the functor comes equipped with an order) where they amount to different names for lax/oplax morphisms, for the sake of making the connection to~\cite{Hasuo06} explicit.

\begin{definition}\label{def:coalg-sims}
    A \emph{forward simulation} from $(X,c)$ to $(Y,d)$ is a lax homomorphism from $(Y,d)$ to $(X,c)$. A \emph{backward simulation} from $(X,c)$ to $(Y,d)$ is an oplax homomorphism. Combining the lax-commuting boxes leads to \emph{hybrid, forward-backward and backward-forward, simulations}. In particular, Hughes-Jacobs simulation is a forward-backward simulation: It is a relation $R$ on which there exists a coalgebra structure $b\colon R \to \cD R$ making $\pi_1$ a forward simulation from $(X,c)$ to $(R,b)$ and $\pi_2$ a backward simulation from $(R,b)$ to $(Y,d)$, depicted:
    \[
\xymatrix@R=1.5em{
X
\ar[d]_{c}
\ar@{}[dr]|{\sqsubseteq}
&
R
\ar[l]_{\pi_1}
\ar[r]^{\pi_2}
\ar[d]_{b}
\ar@{}[dr]|{\sqsubseteq}
&
Y
\ar[d]^{d}
\\
\cD X
&
\cD R
\ar[l]^{\cD\pi_1}
\ar[r]_{\cD\pi_2}
&
\cD Y}
\]
\end{definition}

Clearly, a symmetric simulation on a single system is a bisimulation. 

In the rest of this section, we will show that $\epsilon$-(bi)simulations can be depicted similarly. For this we will need bounded-lax-commutativity of morphisms as well as a notion of $\epsilon$-coupling (which then directly gives an $\epsilon$-relation-lifting). We will focus here on the functor $\cD$ only. These notions can be generalized to generic coalgebras for functors with suitable structure. However, those observations are beyond the scope of this paper and we leave them for future work.

\begin{definition}\label{def:eps-lax-hom}
  An \emph{$\epsilon$-lax homomorphism} from $(X,c)$ to $(Y,d)$ is a morphism $l\colon X \to Y$ that makes the left $\epsilon$-lax diagram below commute, i.e., $d \circ l \sqsubseteq_\varepsilon \cD l\circ c$. 
  An \emph{$\epsilon$-oplax homomorphism} from $(X,c)$ to $(Y,d)$ is a morphism $o\colon X \to Y$ with the property that $\cD o \circ c \sqsubseteq_\varepsilon d\circ o$, i.e., it makes the right $\epsilon$-lax diagram below commute. 
  \[
{\xymatrix@R=1.5em{
X
\ar[r]^{l}
\ar[d]_{c}
\ar@{}[dr]|{_\varepsilon\sqsupseteq}
&
Y
\ar[d]^{d}
&&
X
\ar[r]^{o}
\ar[d]_{c}
\ar@{}[dr]|{\sqsubseteq_\varepsilon}
&
Y
\ar[d]^{d}
\\
\cD X
\ar[r]_{\cD l}
&
\cD Y
&&
\cD X
\ar[r]_{\cD o}
&
\cD Y}
}
\]
The \emph{$\epsilon$-order on distributions} is defined by: For a set $A$ and $\varphi, \psi \in \cD A$, we have $\varphi \sqsubseteq_\varepsilon \psi$ iff for all $B \subseteq A$, $\varphi(B) \le \psi(B) + \varepsilon$.  
\end{definition}

\begin{remark}\label{rem:}
    Note that $\sqsubseteq_\varepsilon$ is not  an order relation, namely it is not transitive. Note also that the definition of $\sqsubseteq_\varepsilon$ can not be expressed on elements (of $A$) only, as it is stronger than the property: for all $a\in A$, $\varphi(a) \le \psi(a) + \varepsilon$. Namely, the role of $\varepsilon$ here, just like in the definition of $\epsilon$-(bi)simulation is \emph{global}. 
\end{remark}

We next recall the notion of $\epsilon$-coupling from~\cite{TraQapl11}, where it is called an $\epsilon$-weight function.

\begin{definition}\label{def:eps-coupling}
Let $R \subseteq X \times Y$. A distribution $\beta \in \cD R$ is an $\epsilon$-coupling for $\mu \in \cD X$ and $\nu \in \cD Y$ iff the following three conditions hold:
\begin{itemize}
\item[1.] $\sum_{y \in Y} \beta(x,y) \le \mu(x)$, for all $x \in X$
\item[2.] $\sum_{x \in X} \beta(x,y) \le \nu(y)$, for all $y \in Y$
\item[3.] $\mu(X) \le \sum_{x \in X, y\in Y} \beta(x,y) + \varepsilon = \sum_{(x,y) \in R} \beta(x,y) + \varepsilon$.
\end{itemize}
\end{definition}

To be fully precise, we should consider $\beta \in \cD (X \times Y)$ with support contained in $R$. We write $\beta \in \cD R$ and yet sometimes sum over all $x\in X, y \in Y$, i.e., we identify $\beta \in \cD R$ with a distribution in $\cD (X\times Y)$ that assigns probability $0$ to all pairs out of $R$ and acts as $\beta$ on $R$.

Before we proceed with the main observation of this section, we prove an auxiliary property that allows for rewriting the definition of $\epsilon$-coupling.

\begin{lemma}\label{lem:eps-coupl-aux}
    Assume that Condition 1.~in Definition~\ref{def:eps-coupling} holds for $\beta \in \cD R$ with $R \subseteq X \times Y$ and $\mu \in \cD X$. Then $\mu(X) \le \sum_{x \in X, y\in Y} \beta(x,y) + \varepsilon$ (i.e., Condition 3.~in Definition~\ref{def:eps-coupling}) is equivalent to the condition
    \[
    \mu(A) \le \sum_{x \in A, y\in Y} \beta(x,y) + \varepsilon, \quad \text{for all } A \subseteq X.
    \]
\end{lemma}

\begin{proof}
    The right-to-left implication is immediate, as $X \subseteq X$. For the left-to-right implication, assume that for some $A\subseteq X$ we have 
    \[
    \mu(A) > \sum_{x \in A, y\in Y} \beta(x,y) + \varepsilon.
    \]
    Since, by assumption, for all $x \in X$, and hence for all $x \in X\setminus A$, $\mu(x) \ge \sum_{y \in Y} \beta(x,y)$, we have
    \begin{eqnarray*}
        \mu(X)  &=& \mu(A) + \mu(X\setminus A)\\
        &>& \sum_{x \in A, y\in Y} \beta(x,y) + \varepsilon + \sum_{x \in X \setminus A} \mu(x)\\
        &\ge& \sum_{x \in A, y\in Y} \beta(x,y) + \sum_{x \in X \setminus A} \sum_{y \in Y} \beta(x,y) + \varepsilon\\
        &=& \sum_{x\in X,y\in Y}\beta(x,y) + \varepsilon.
    \end{eqnarray*}
    The property now follows by contraposition. 
\end{proof}
This property again emphasises the global nature of $\varepsilon$ in our situation. As a consequence, for an $\epsilon$-coupling $\beta \in \cD R$ of $\mu \in \cD X$ and $\nu \in \cD Y$, with $R \subseteq X\times Y$, for each $A \subseteq X$: 
\[
\sum_{x\in A, y\in Y} \beta(x,y)\,\, \le\,\, \mu(A) \,\,\le \,\,\sum_{x\in A, y\in Y} \beta(x,y) + \varepsilon .
\]

\begin{proposition}\label{prop:eps-coal-bisim}
Let $(X,c)$ and $(Y,d)$ be two $\cD$-coalgebras. The following three properties are equivalent for $R \subseteq X\times Y$:
\begin{itemize}
    \item[1.] $R$ is an $\epsilon$-simulation.
    \item[2.] For every $(x,y) \in R$, there is an $\epsilon$-coupling $\beta\in \cD R$ of $\mu = c(x)$ and $\nu = d(y)$.
    \item[3.] The following "$\epsilon$-lax-bounded" span of morphisms commutes:
        \[
\xymatrix@R=1.5em{
X
\ar[d]_{c}
\ar@{}[dr]|{\sqsupseteq\, \sqsubseteq_\varepsilon}
&
R
\ar[l]_{\pi_1}
\ar[r]^{\pi_2}
\ar[d]_{b}
\ar@{}[dr]|{\sqsubseteq}
&
Y
\ar[d]^{d}
\\
\cD X
&
\cD R
\ar[l]^{\cD\pi_1}
\ar[r]_{\cD\pi_2}
&
\cD Y}
\]
\end{itemize}  
\end{proposition}

\begin{proof}
    The equivalence of 1. and 2. has been shown in~\cite{TraQapl11}. Let $R$ be an $\epsilon$-simulation. For $(x,y) \in R$, let $\mu = c(x)$ and $\nu = d(y)$. Define $b\colon R \to \cD R$ by $(x,y) \mapsto \beta$, the epsilon coupling of $\mu$ and $\nu$ that exists.  

    Then, unfolding the definitions, we get:
    \begin{itemize}
        \item $\cD \pi_1 \circ b \sqsubseteq c \circ \pi_1$ is equivalent to: for all $(x,y) \in R$, condition 1. from Definition~\ref{def:eps-coupling} holds.
        \item $\cD \pi_2 \circ b \sqsubseteq d \circ \pi_2$ is equivalent to: for all $(x,y) \in R$, condition 2. from Definition~\ref{def:eps-coupling} holds.
        \item $ c \circ \pi_1\sqsubseteq_\varepsilon \cD \pi_1 \circ b$ is equivalent to: for all $(x,y) \in R$, for all $A \subseteq X$, \[\mu(A) \le \sum_{x\in A,y\in Y} \beta(x,y) + \varepsilon\] which by Lemma~\ref{lem:eps-coupl-aux} is equivalent to condition 3. from Definition~\ref{def:eps-coupling}.
    \end{itemize}
    These facts yield the equivalence of 3. with 2. (and hence 1. as well). 
\end{proof}

Note that Condition 2. in Proposition~\ref{prop:eps-coal-bisim} can be taken as a definition of $\epsilon$-relation-lifting, in analogy to relation lifting for $\cD$ being defined using the existence of a coupling for any pair of elements in $R$. This way was also taken in~\cite{HughesJacobs04} with the definition of lax relation lifting. 

As a consequence of Proposition~\ref{prop:eps-coal-bisim}, we immediately get that a relation $R \subseteq X\times X$ on the states of a $\cD$-coalgebras is an $\epsilon$-bisimulation iff the following diagram commutes: 
\[
\xymatrix@R=1.5em{
X
\ar[d]_{c}
\ar@{}[dr]|{\sqsupseteq\, \sqsubseteq_\varepsilon}
&
R
\ar[l]_{\pi_1}
\ar[r]^{\pi_2}
\ar[d]_{b}
\ar@{}[dr]|{_\varepsilon\sqsupseteq\, \sqsubseteq}
&
X
\ar[d]^{d}
\\
\cD X
&
\cD R
\ar[l]^{\cD\pi_1}
\ar[r]_{\cD\pi_2}
&
\cD X}
\]
This diagram resembles the defining diagram of the recently introduced uncertain bisimulations~\cite{DBLP:conf/calco/RotW23}, but requires a double property (for $\sqsubseteq_\varepsilon$ as well as for $\sqsubseteq^{op}$). A detailed comparison with this interesting work remains a future task.

As already mentioned, generalizing the notions of approximate (bi)simulations to generic coalgebras is an interesting direction for future work. A similar notion is being developed for cost automata~\cite{PedroMayukoKoko26}. Another very recent work building on our results gives a generic notion of threshold-based equivalences, based on novel kind of predicate liftings~\cite{DBLP:conf/concur/ForsterSWKN26}.

\section{Concluding remarks}\label{sec:conc}

We have shown that $\epsilon$-bisimulations are closely connected with the Lévy-Prokhorov pseudometric on (sub)probability distributions: The LP pseudometric yields a lifting of $\cD$ to a functor on pseudometric spaces and induces coinductive behaviour distance on LMPs,  which is the greatest fixpoint of a suitable functional and turns out to be exactly the distance induced by $\epsilon$-bisimilarity. Remarkably, any (post)fixpoint distance of that functional defines an $\epsilon$-bisimulation.
This is the first time that a distance on distributions other than the Kantorovich distance is used for characterizing behavioural distances. 

\bibliographystyle{alphaurl}

\bibliography{main.bib}

\end{document}